\documentclass[a4paper,11pt]{article}
\usepackage[bbgreekl]{mathbbol}
\usepackage[small]{titlesec}
\usepackage[utf8]{inputenc}
\usepackage[T1]{fontenc}  
\usepackage[backend=biber,style=numeric-comp,natbib=false,giveninits=true,doi=true,isbn=false,url=false,date=year,sorting=nyvt]{biblatex}
\DeclareNameAlias{default}{family-given}

\DeclareFieldFormat[article]{title}{#1}
\renewbibmacro{in:}{}
\DeclareFieldFormat[article]{pages}{#1}
\DeclareFieldFormat{journal}{#1}
\renewcommand\bf\bfseries

\renewbibmacro*{volume+number+eid}{%
	\printfield{volume}%
	\setunit*{\addnbspace}
	\printfield{number}%
	\setunit{\addcomma\space}%
	\printfield{eid}}
\DeclareFieldFormat[article]{number}{\mkbibparens{#1}}
\DeclareFieldFormat[article]{volume}{\textbf{#1}}
\DeclareFieldFormat{year}{\mkbibparens{#1}}
\DeclareBibliographyDriver{article}{%
	\printnames{author}:%
	\newunit\newblock
	\printfield{title}%
	\newunit\newblock
	\printfield{journaltitle}
	\newunit
	\iffieldundef{number}{\printfield{volume}}{\printfield{volume}\addspace\printfield{number}}
	\addcomma\addspace\printfield{pages}\addspace
	\printfield{year}
}
\AtEveryBibitem{\clearfield{month}}
\AtEveryBibitem{\clearfield{day}}
\usepackage[english]{babel}
\usepackage[T1]{fontenc}
\usepackage{csquotes}
\usepackage{bbm}
\usepackage{amsmath,amsfonts,amsthm,amsbsy,amssymb,dsfont,stmaryrd}
\usepackage[dvipsnames]{xcolor}
\usepackage{braket}

\usepackage{mathtools}

\numberwithin{equation}{section}

\newcommand\myshade{85}
\colorlet{mylinkcolor}{violet}
\colorlet{mycitecolor}{YellowOrange}
\colorlet{myurlcolor}{Aquamarine}

\usepackage[left=2.5cm,right=2.5cm,top=2.5cm,bottom=2.5cm]{geometry}
\usepackage[unicode=true,pdfusetitle,
bookmarks=true,bookmarksnumbered=false,bookmarksopen=false,
breaklinks=false,pdfborder={0 0 1},backref=false,
linkcolor  = mylinkcolor!\myshade!black,
citecolor  = mycitecolor!\myshade!black,
urlcolor   = myurlcolor!\myshade!black,
colorlinks = true,
]
{hyperref}
\usepackage[nameinlink]{cleveref}

\usepackage{tikz}
\usetikzlibrary{arrows}
\usetikzlibrary{intersections}

\usepackage{graphicx}
\usepackage{caption}

\definecolor{ct_black}{HTML}{000000}
\definecolor{ct_orange}{HTML}{ED872D}
\definecolor{ct_purple}{HTML}{7A68A6}
\definecolor{ct_blue}{HTML}{348ABD}
\definecolor{ct_turquoise}{HTML}{188487}
\definecolor{ct_red}{HTML}{E32636}
\definecolor{ct_pink}{HTML}{CF4457}
\definecolor{ct_green}{HTML}{467821}

\definecolor{ct2_green}{HTML}{9FF781}
\definecolor{ct2_green_dark}{HTML}{088A08}

\theoremstyle{plain}
\newtheorem{thm}{\protect\theoremname}[section]
\theoremstyle{plain}
\newtheorem{lem}[thm]{\protect\lemmaname}
\theoremstyle{plain}

\theoremstyle{plain}
\newtheorem{prop}[thm]{\protect\propositionname}
\theoremstyle{remark}

\newtheorem{assumption}[thm]{\protect\assumptionname}

\theoremstyle{remark}
\newtheorem{rem}[thm]{\protect\remarkname}
\theoremstyle{definition}

\theoremstyle{plain}

\providecommand{\assumptionname}{Assumption}
\providecommand{\claimname}{Claim}
\providecommand{\corollaryname}{Corollary}
\providecommand{\definitionname}{Definition}
\providecommand{\lemmaname}{Lemma}
\providecommand{\propositionname}{Proposition}
\providecommand{\remarkname}{Remark}
\providecommand{\theoremname}{Theorem}
\providecommand{\examplename}{Example}

\crefname{section}{Section}{Sections}
\crefname{appendix}{Appendix}{Appendices}
\crefname{figure}{Figure}{Figures}
\crefname{assumption}{Assumption}{Assumptions}
\crefname{thm}{Theorem}{Theorems}
\crefname{lem}{Lemma}{Lemmas}
\crefformat{equation}{(#2#1#3)}

\crefrangelabelformat{equation}{(#3#1#4--#5#2#6)}

\crefmultiformat{equation}{(#2#1#3}{, #2#1#3)}{#2#1#3}{#2#1#3}
\Crefmultiformat{equation}{(#2#1#3}{, #2#1#3)}{#2#1#3}{#2#1#3}

\newtheorem*{lem*}{\protect\lemmaname}

\newcommand{\ee}{\operatorname{e}}
\newcommand{\ii}{\operatorname{i}}

\newcommand{\RR}{\mathbb{R}}
\newcommand{\CC}{\mathbb{C}}

\newcommand{\EE}{\mathbb{E}}

\newcommand{\norm}[1]{\left\|#1\right\|}
\newcommand{\ip}[2]{\langle #1, #2 \rangle}

\newcommand{\dif}{\operatorname{d}}
\newcommand{\tr}{\operatorname{tr}}

\renewcommand{\Re}{\operatorname{\mathbb{R}\mathbbm{e}}}

\newcommand{\vf}{\varphi}
\newcommand{\Id}{\mathds{1}}
\newcommand{\HH}{\mathcal{H}}

\newcommand{\supp}{\operatorname{supp}}

\newcommand{\VV}{\mathcal{V}}

\newcommand{\vm}{|v_{min}|}
\newcommand{\vmh}{\nu}

\usepackage{environ}

\NewEnviron{malign}{%
	\begin{align}\begin{split}
			\BODY
	\end{split}\end{align}
}

\title{Lower Bound on Quantum Tunneling for Strong Magnetic Fields}
\author{Charles L. Fefferman\\
	{\footnotesize Department of Mathematics, Princeton University}
	\\Jacob Shapiro\\
	{\footnotesize Department of Physics, Princeton University}\\
	 Michael I. Weinstein\\
		\footnotesize{Department of Applied Physics and Applied Mathematics, and Department of Mathematics, Columbia University}
}

\begin{document}
	
\maketitle

\begin{abstract}
We consider a particle bound to a two-dimensional plane and a double well potential, subject to a perpendicular uniform magnetic field . The energy difference between the lowest two eigenvalues--the eigenvalue splitting--is related to the tunneling probability between the two wells. We obtain upper and lower bounds on this splitting  in the regime where both the magnetic field strength and the depth of the wells are large. The main step is a lower bound on the hopping amplitude between the wells, a key parameter in tight binding models of solid state physics, given by an oscillatory integral, whose phase has no critical point, and which is exponentially small.
\end{abstract}
 
\section{Introduction}
The quantum double-well exhibits the basic purely quantum phenomenon of tunneling. The double-well Hamiltonian  is 
\begin{equation} P^2+\lambda^2 V(x),\quad 
\label{HnoB}\end{equation}
 where $P\equiv\frac{1}i\ \nabla$ denotes the momentum operator, and 
 \begin{equation} V(x)\ =\ v_0(x)\ + v_d(x)\ \equiv\ v_0(x-0)\ +\ v_0(x-d).\label{Vdw}
 \end{equation}
 Here, $v_0(x)\le0$ is taken to be radially symmetric function on $\RR^2$, with a strictly negative minimum about $x=0$,
and decaying to zero as $|x|$ tends to infinity.  
 Thus, $V(x)$ is a double-well with two distinct minima displaced from one another by a vector $d\in\RR^2$ and $\lambda>0$ is the coupling constant. In the classical mechanics system with potential $V$, the ground (least energy) state is doubly degenerate, with contributions coming from each of the two wells. In contrast, for any fixed $d$ and $\lambda$ sufficiently large, the quantum system has a nearly degenerate pair of eigenvalues $\mu^\lambda_0<\mu^\lambda_1$.
The ground state eigenfunction with energy $\mu^\lambda_0$ is approximately given by an \emph{even} linear superposition of the ground state for a Hamiltonian with an isolated well $v_0(x)$, centered at $x=0$, and that of the Hamiltonian with an isolated well $v_d(x)$, centered at $d$. The excited state with energy $\mu^\lambda_1$ is approximately given by an \emph{odd} superposition of these states. The reciprocal of the {\it eigenvalue splitting} $\mu^\lambda_1-\mu^\lambda_0$ is a measure of the time it takes for a particle to tunnel from one well to the other. If $v_0$ has compact support of radius smaller than $|d|/2$, the splitting satisfies
  \begin{equation}
  \exp(-c_1 \lambda |d|)\ <\ \mu^\lambda_1-\mu^\lambda_0\ <\ \exp(-c_2 \lambda |d|)\ \ \textrm{for $\lambda$ sufficiently large}\,,
  \label{splitting}\end{equation}
  where $c_1>c_2$ are constants which depend on $v_0$ and $d$ but not on $\lambda$; see also \cref{eq:non-magnetic-splitting} below for other cases. 
  The calculation of the splitting is a well-known consequence of the WKB approximation \cite[Ch. 50, problem 3]{Landau_Lifshitz_vol_3}, and can also be obtained via the dilute instanton gas approximation \cite{Coleman_aspects_of_symmetry}. Rigorous calculations of the splitting (in arbitrary spatial dimension), assuming $v_0$ has a non-degenerate minimum, appeared in \cite{Simon_1984_10.2307/2007072} using either Agmon estimates or by a Laplace type asymptotics applied to the Feynman-Kac representation of the heat kernel \cite{Schilder_66_10.2307/1994588}, as well the semiclassical techniques of  \cite{Helffer_Sjostrand_1984}; see also references cited therein. The main technical challenge in these works is a \emph{lower bound} on the tunneling probability, where in all techniques
the positivity of the single well ground state plays a crucial role.
\\

 Consider now the question of a quantum particle in the two-dimensional plane confined by a  double well interacting with a perpendicular magnetic field. 
 \medskip
 
\noindent {\it Question: What is the nature of the splitting in the simultaneous limit of deep potential wells and strong magnetic field?}
\medskip

In particular, we consider the case of deep potential wells and a strong \emph{constant} magnetic field 
described by the  magnetic double-well Hamiltonian
\begin{align} H^{\lambda,b} := (P-b A x)^2+\lambda^2V(x),\quad   \lambda,b\gg1.
\label{eq:two-well Hamiltonian}\end{align} 
Here,  $V$ denotes a double well potential of the form given in \eqref{Vdw} and  $Ax$ is the vector potential in the symmetric gauge for a constant magnetic field:
 \[ Ax\ =\  \frac{1}{2} \ e_3 \wedge x=\frac12\ (x_2,-x_1) ,\quad B(x)\ =\ {\rm curl}\ b Ax\ =\ b\ e_3.\] 
 
Specifically, in this paper we assume $b$ scales with $\lambda$ and that $\lambda$ is large; in particular we make the following:
\begin{assumption}
	\label{ass:magnetic field strength scaling like potential well}
	The magnetic parameter $b$ scales like the coupling $\lambda$,\ $b=\lambda$.
\end{assumption}
\begin{rem}[Semiclassical regime] Set $\hbar=\lambda^{-1}$. Then, since
\[ H^{\lambda,\lambda} = \hbar^{-2}\Big[\ \left( \frac{\hbar}{i}\nabla - Ax \right)^2+ V(x)  \Big] ,\] 
the limit $\lambda\to\infty$ is directly related to the semiclassical limit $\hbar\to0$.
\end{rem}
The choice of scaling in \cref{ass:magnetic field strength scaling like potential well} is motivated by our application to continuum models of the integer quantum Hall effect (IQHE) in a crystal which inherit topologically non-trivial sub-band structures of a tight binding limiting Hamiltonian. See  \cite{Shapiro_Weinstein_2020} where, as an example, the Harper model \cite{PhysRevB.14.2239} is realized in this limit, and \cref{tb-red} below. Another topologically non-trivial regime arises for $\lambda$ fixed and $b$ tending to infinity; in this limit,  the Landau Hamiltonian dominates and the effect of the crystal potential washes out \cite{Bellissard:1987dy}. In contrast, the regime of fixed magnetic field strength, $b$, and  strong binding ($\lambda$ large) was shown in \cite{nakamura1990} to be topologically trivial in the limit.

Further, our results are of independent mathematical interest. Previous semi-classical approaches to the lower bound on the magnetic eigenvalue splitting and magnetic hopping were obtained only in the perturbative regime of $b$ small \cite{Helffer_Sjostrand_1987_magnetic_ASNSP_1987_4_14_4_625_0}. Moreover, it is not clear how to apply WKB asymptotics (which were successfully applied in the non-magnetic case \cite{Helffer_Sjostrand_1984}), to treat
 the current problem since the associated eikonal equation here is complex valued.

Finally, the particular $\lambda=b$ scaling, one obtains a natural comparison with the scaled magnetic harmonic oscillator, if the potential has a unique minimum, as in \cite{MATSUMOTO1995168}. 

\begin{rem}
	In this paper, $\lambda$ is a dimensionless asymptotic parameter, and we use the natural units $\hbar = c = 2m_e = q_e = 1$ ($m_e,q_e$ being the mass and respectively the charge of the electron, $c$ the speed of light). Hence, the equation $b= \lambda$, valid only in natural units, is equivalent to the universal $b = 4\frac{m_e^2 c^2}{q_e \hbar} \lambda$ and similarly throughout. We note that the objects $P^2,X^{-2}$ and $b$ all have length scaling $L^{-2}$.
\end{rem}

Let  $(\vf^{\lambda,b}_0(x), e_0^{\lambda,b})$ denote the $L^2-$ normalized ground state eigenpair
of the \emph{single} well magnetic  Hamiltonian:
\begin{equation}
 h^{\lambda,b} := (P -\ b A x)^2+\lambda^2v_0(x) \label{eq:the one-well Hamiltonian}\,.
 \end{equation}

We make the following assumption on the atomic potential well, $v_0 \leq 0$:
\begin{assumption}\label{ass:radial monotone cpt supp for v}
\begin{enumerate}
	\item[(v1)] $v_0(x)=v_0(r)$, $r=|x|$, is bounded (we denote its minimum by $v_\mathrm{min} < 0$), radial, and compactly supported: $\supp (v_0) \subset B_a(0)$ for some $a>0$.
	\item[(v2)] $\vf^{\lambda,b}_0(x)$, the ground state derived from $v_0$ via \cref{eq:the one-well Hamiltonian}, is a radial function.
	\item[(v3)] For $b=\lambda$, $h^{\lambda,\lambda}$ has a spectral gap of order at least $1$ in $\lambda$ above $e_0^{\lambda,\lambda}$, i.e., if $e_1^{\lambda,\lambda}$ is the energy for the first excited state then \begin{align} \liminf_{\lambda\to\infty}\left(e_1^{\lambda,\lambda}-e_0^{\lambda,\lambda}\right)>0\,. \label{eq:spectral gap for atomic Hamiltonian}\end{align}
\end{enumerate}
\end{assumption}

We believe that (v2) and (v3) hold for a large class of radial potentials, $v_0(r)$ and $\lambda$ sufficiently large, {\it e.g.} $v_0(r)$ is sufficiently smooth with a non-degenerate minimum at $r=0$. We plan to address this in detail in a future article. One could also come up with examples of radial potentials for which this condition is violated, e.g. a potential shaped like a Mexican hat. The problem of obtaining analogs of our results for nonradial $v_0$ remains open.

For $V\equiv0$, $H^{\lambda,b}= H^{0,b}$ is the classical Landau Hamiltonian whose spectrum consists of an infinite sequence
of discrete eigenvalues tending to infinity in arithmetic progression. Each eigenvalue, a Landau level, is of infinite multiplicity and belongs to the essential spectrum. The essential spectra of $ h^{\lambda,b}$ and $H^{\lambda,b}$ for non-zero potentials are unaffected as the potential is a relatively compact perturbation, but the discrete spectrum includes now sequences of finite multiplicity eigenvalues which converge to the Landau levels. The rate of convergence of these eigenvalues, as it depends on the decay properties of $v_0$, was studied in \cite{Raikov_Warzel_2002}.
  
Since the double well potential $V$ is non-positive, the sequence of eigenvalues of $H^{\lambda,b}$ which approaches the first Landau level converges to it from below.

{\it Our goal is to prove upper and lower bounds on  the splitting of the lowest two eigenvalues in this sequence,  $E^{\lambda,\lambda}_1-E^{\lambda,\lambda}_0$.}

A notational convention: henceforth, instead of writing $H^{\lambda,\lambda}$ and $E^{\lambda,\lambda}$ we simply write $H^\lambda$ and $E^{\lambda}$. Sometimes, to emphasize the dependence on the magnetic field,  we restore $b$ in the notation and revert to  $H^{\lambda,b}$ and $E^{\lambda,b}$. 

\begin{thm}[Bounds on magnetic eigenvalue splitting]\label{thm-split} Consider the magnetic double well Hamiltonian $H^{\lambda,b}$, \cref{eq:two-well Hamiltonian}, where the double well potential, $V$,  consists of identical compactly supported atomic wells satisfying \cref{ass:radial monotone cpt supp for v}, {\it i.e.} each atomic well is supported in a disc of radius $a>0$, with well-separation vector $d\in\RR^2$; see \eqref{Vdw}. We further assume that 

\begin{align}  
 |d|>4(\sqrt{\vm}+a)\,. \label{eq:spacing}
\end{align}
	Then, there exist positive constants: $\lambda_\star,C$ depending on $v_0,d$ and $\gamma_0$ depending on $v_0$ such that for all $\lambda\ge\lambda_\star$ the
	 separation  between the two lowest eigenvalues of $H^\lambda$ satisfies the upper and lower bounds:
	 \begin{align} \exp\left(-\frac{1}{4}\lambda\left(|d|^{2}+4\sqrt{\vm}|d|+\gamma_0\right)\right)\leq  E^{\lambda}_1-E^{\lambda}_0 \leq 
	  C\lambda^{5/2}\exp\left(-\frac{1}{4}\lambda\left((|d|-a)^2-a^2\right)\right)\,. \label{eq:main result} \end{align}
\end{thm}
\begin{rem}\label{rem:wells separated}
	We note in particular that \cref{eq:spacing} implies that \begin{align}|d|>2a\,,\label{eq:minimal spacing between wells}\end{align} i.e., in our setting the supports of the two wells are disjoint.
\end{rem}

An  \emph{upper} bound on the splitting was obtained in \cite{Nakamura1999}. An asymptotic expansion 
 of the splitting for the case of weak magnetic fields was studied in \cite{Helffer_Sjostrand_1987_magnetic_ASNSP_1987_4_14_4_625_0}.
  In \cite{Bruning2013} a class of  separable potentials, facilitating a  reduction to the non-magnetic case,
 was treated. The class of potentials we treat is not covered by this work. It appears that our result is the first for a physically interesting class of potentials.
 
  At the heart of the proof of \cref{thm-split} is a lower bound on the magnitude of the \emph{magnetic hopping coefficient}, whose absolute value squared is the probability of an electron hopping, in the presence of a magnetic field, from the well centered at $x=0$ to the well centered at $x=d$: 
   \begin{align}
	\rho^{\lambda,b}\left(d\right) \ &\equiv\ \left\langle \vf_0^{\lambda,b},\ \lambda^2 v_0\ \vf_{d}^{\lambda,b}\right\rangle_{L^2(\RR^2)}\nonumber\\
	& =  \lambda^{2}\int_{\RR^2}\overline{\vf_0^{\lambda,b}\left(x\right)}\ v_0\left(x\right)\ \exp\left(\ii x\cdot b Ad\right)\vf_0^{\lambda,b}\left(x-d\right)\dif{x}\label{m-hop}\end{align}
	The magnetic hopping coefficient, $\rho^{\lambda,b}\left(d\right) $, is  an overlap integral involving $\vf_0^{\lambda,b}$, the ground state centered at $x=0$, 
	$v_0$, the atomic potential centered at $x=0$, and $\vf_{d}^{\lambda,b}(x) =  e^{\ii x\cdot\lambda Ad}\vf_0^{\lambda,b}\left(x-d\right)$, the {\it magnetic translate} of the ground state $\vf_0$ 
	 to the site $x=d$. The notion of magnetic translation is discussed in detail in \cref{sec:m-trans}. Clearly $\rho^{\lambda,b}\left(d\right) $ is a-priori a complex-valued quantity  and so lower bounds are necessarily quite subtle. Indeed the proofs of the corresponding lower bound for the non-magnetic case ($b=0$), e.g., \cite[Prop. 4.1]{FLW17_doi:10.1002/cpa.21735}, make use of the positivity of the integrand. In contrast, obtaining a lower bound for the magnetic case is much more subtle, and required a new approach; indeed $\rho^\lambda(d)$, given by $\cref{m-hop}$, is an  oscillatory integral; the integrand involves the complex valued magnetic-translate of $\vf_0^\lambda$ and the phase function of the integral has no critical points; precision asymptotic expansions of $\vf_0^\lambda$ do not address the problem at hand.
	 
	 We state the lower bound for $b=\lambda$ large:
	 
\begin{thm}[Bounds on hopping coefficient]\label{lb-hop}
Under \cref{ass:radial monotone cpt supp for v,ass:magnetic field strength scaling like potential well} as well as assuming \cref{eq:minimal spacing between wells}, there exist positive constants: $\lambda_\star,C$ depending on $v_0,d$ and $\gamma_0$ depending on $v_0$ such that for all $\lambda\ge\lambda_\star$,
\begin{equation}  \exp\left(-\frac{1}{4}\lambda\left(|d|^{2}+4\sqrt{\vm}|d|+\gamma_0\right)\right)\leq\left|\rho^{\lambda}(d)\right|\leq C\lambda^{5/2}\exp\left(-\frac{1}{4}\lambda((|d|-a)^2-a^2)\right)\,.\label{rho-bd}
\end{equation}
\end{thm}	 

\cref{lb-hop} plays an important role in the derivation of tight binding models for continuum magnetic two-dimensional crystalline (not necessarily translation invariant) structures \cite{Shapiro_Weinstein_2020}; see the discussion in \cref{tb-red}. In fact, for us, \cref{lb-hop} is a precursor of \cref{thm-split} in the sense that we prove the latter by proving the former together with the key identity \begin{align}\lim_{\lambda\to\infty}\left|\frac{E^{\lambda}_1-E^{\lambda}_0}{2|\rho^{\lambda}(d)|}-1\right|=0\,.\label{eq:key identity for eignevalue splitting}\end{align}

\begin{rem}
In our proof we make use of the positivity of $\vf_0^{\lambda,b}$ which may be arranged WLOG under the second item of \cref{ass:radial monotone cpt supp for v}.  \end{rem}

\begin{rem}
	In \cref{eq:two-well Hamiltonian} we have used the symmetric gauge, which implements a uniform magnetic field of strength $b$ perpendicular to the plane. The choice of gauge should of course not matter for the analysis but we make it here for convenience as it facilitates dealing with the radial potential $v_0$.
\end{rem}

\begin{rem}
	A purely magnetic ($V\equiv0$) tunneling effect has recently been studied in \cite{1912.04035}, where the decisive step is reduction to a non-magnetic tunneling problem \cite[eq-n (1.3)]{1912.04035}. Here, in contrast, we are interested in the tunneling amplitude between two atoms immersed in a constant perpendicular magnetic field. Our lower bound on $\rho^\lambda$ is unrelated to the non-magnetic problem due to the oscillatory integral over $\RR^2$, which contributes to further exponential dampening and is not amenable to a saddle point approximation.
\end{rem}
In concluding this section, we quantitatively comment on the dampening effect of the oscillatory phase in \cref{m-hop}. If we had used ordinary translations instead of magnetic translations in \cref{m-hop} (which of course yields an irrelevant object for the construction of an approximate spectral basis for the two-atom Hamiltonian) then the oscillatory phase in \cref{m-hop} would be absent but otherwise all other expressions are identical. The upper bound on $\rho^\lambda(d)$ is unaffected (as that didn't take the phase into account) but the lower bound would be larger by a factor $$ \exp(\frac{1}{4}\lambda a^2)\,. $$ Indeed, to adjust for this change, one would now obtain $I_0(\lambda|d|r(\frac{1}{2}+t))$, instead of \linebreak$I_0(\lambda|d|r\sqrt{t(1+t)})$, on the RHS of \cref{eq:the result of the Bessel integral}. Clearly, we already see that $\frac{1}{2}+t$ is strictly larger than the result with oscillations $\sqrt{t(1+t)}$. Proceeding from \cref{eq:the result of the Bessel integral} in the same way as we do in the proof below, we ultimately find \cref{eq:lower bound on hopping coefficient with explicit normalization constant} modified by removing the term $a^2$ in the exponential and hence the claim. We note that the comparison to the non-magnetic estimates of $\left.|\rho^{b,\lambda}(d)|\right|_{b=0}$, which can be read off from \cite[Section 15.3]{FLW17_doi:10.1002/cpa.21735}, is not very informative because the non-magnetic expression depends on the shape of the potential whereas in the magnetic case the dominant behavior is determined by the Gaussian decay (in $|d|$) caused purely by the magnetic field. 

\subsection{Remarks on possible probabilistic approaches}\label{sec:prob approach}

	In the non-magnetic case, $b=0$, it is proven in  \cite{Simon_1984_10.2307/2007072} that (using the notation $\mu$ for eigenvalues of the non-magnetic system) that the eigenvalue splitting satisfies
	 \begin{align} \lim_{\lambda\to\infty}-\frac{1}{\lambda}\ 
	 \log(\mu^\lambda_1-\mu^\lambda_0)\ = \mathcal{S}(0,d)\,. \label{eq:non-magnetic-splitting}
	 \end{align}
	 Here $\mathcal{S}(0,d)$ is the extremal Euclidean action from the origin to $d$:
	  $$ S_T(\gamma) = \int_0^T \frac{1}{2}\ |\dot{\gamma}|^2 + V\circ \gamma$$ 
	 \footnote{Note that the sign of the potential is opposite that of the usual Lagrangian of mechanics. Additionally, the asymptotic parameter in front of $V$ is absent.} evaluated on the classical solution to the equations of motion corresponding to $S_T$, with boundary conditions $\gamma(0) = 0,\gamma(T) = d$, that is $$ \mathcal{S}(0,d) \equiv \inf\{S_T(\gamma):\ \gamma:[0,T]\to\RR^2:\gamma(0) = 0,\gamma(T) = d\} . $$ 
	  $\mathcal{S}(0,d)$ also corresponds to the Agmon distance from the origin to $d$, and indeed in \cite{Simon_1984_10.2307/2007072} it is proven that $$ \mathcal{S}(0,d) = \inf \int_0^1\sqrt{2V \circ \gamma}\ |\dot{\gamma}|,$$ 
	  where the infinimum is taken over all $\gamma:[0,1]\to\RR^2$ such that $\gamma(0) = 0,\gamma(1) = d$.
	
	One might wonder whether it is possible to improve on the upper and lower bounds we derive here on $E^{\lambda}_1-E^{\lambda}_0$ in \cref{thm-split} with a leading order asymptotic formula of the form \cref{eq:non-magnetic-splitting}  in the magnetic case as well. We are not aware that a natural Agmon metric exists in the magnetic case. However, there is a Feynman-Kac-Itô formula for the heat kernel of a magnetic Hamiltonian \cite{barrysimon2004,matsumoto1990} (corresponding to the magnetic Hamiltonian $H^{\lambda,b}$ in \cref{eq:two-well Hamiltonian}): 
	\begin{equation}
	 e^{-t H}(x,y) = (2\pi t)^{-\nu/2}\ee^{-\frac{|x-y|^2}{2t}}\EE_{x\to y}\left[ \exp\left( \int_0^t \ii\frac{b}{2}(e_3 \wedge \gamma )\cdot \dif{\gamma} - \lambda^2 V\circ \gamma\right)\right] \ ,\label{FKI}\end{equation}
	 where $\EE_{x\to y}$ denotes the expectation value of $\nu$-dimensional pinned Brownian motion $\gamma$ obeying $\gamma(0)=x$ and $\gamma(t)=y$. 

	\medskip
	
\noindent {\it Question:}\	{\it Are there stationary phase type arguments for \eqref{FKI} that relate the eigenvalue splitting $E^{\lambda,b}_1-E^{\lambda,b}_0$ with the magnetic action in a manner  analogous  to
 \cref{eq:non-magnetic-splitting} via the classical Euclidean action?}

\subsection{Applications to tight-binding reduction schemes}\label{tb-red}

Consider a crystal which is described as an infinite sum of translates of an atomic well, $v_0$, where the translates may lie in a lattice, {\it e.g.} $\mathbb{Z}^2$ or a more general discrete subset of $\RR^2$.  In the strong binding regime (deep potential wells), it is sometimes possible to rigorously derive an effective, {\it  tight binding}, operator, whose  spectrum gives a good approximation of the low energy spectrum of the crystal, and which also gives information on time-dynamics for initial conditions spectrally localized on this part of the spectrum. 

The tight binding operator is a discrete operator whose coefficients encode the hopping probabilities for an electron transiting from one site to its neighbors. A key ingredient in its derivation is a comparison of the  hopping probability from an atom to its nearest neighbor atom, a distance $|d|$ away, and to its next nearest and more distant neighbors, a distance at least $x|d|>|d|$ away. Here, $x>1$ depends on the details of the crystal lattice; see, for example, \cite{FLW17_doi:10.1002/cpa.21735}. For the square lattice, $x=\sqrt2$, and for the honeycomb lattice, $x=\sqrt3$. 
 \cref{lb-hop} can be used to prove the following comparison of hopping probabilities for appropriate choices of lattice parameters. In \cite{Shapiro_Weinstein_2020}, we use such comparisons to derive tight binding approximations for continuum magnetic two-dimensional crystalline (not necessarily translation invariant) structures in the regime of 
 deep potential wells and strong constant magnetic field.

 \begin{thm}\label{comp-hop} Assume the setting of \cref{lb-hop}; in particular, $|d|>2a$. Let $x>1$ be given. Then, there is a universal constant $C_\star>0$ such that for all $\lambda > 0$,
 \begin{equation} \left|\frac{\rho^\lambda(xd)}{\rho^\lambda(d)}\right|\leq C_\star \exp\left(-\frac18\lambda (x^2-1)|d|^2\right)\label{rho-quot}\,.\end{equation}
\end{thm}
\noindent This result is proved below as \cref{lem:factorize NN out of NNN hopping}. It is a crucial ingredient in the proof of \cite[Proposition 4.3]{Shapiro_Weinstein_2020}.

\subsection{Outline of the paper}\label{sec:outline}
In \cref{sec:single-atom} we study properties of the single atom magnetic ground state, $\vf_0$, which are necessary for our analysis. In \cref{sec:magnetic orbital subspace} we introduce a two-dimensional subspace $\VV$ of $L^2(\RR^2)$,
 the span of the atomic ground state and its magnetic translate, which approximates the eigenspace associated with the lowest two eigenvalues of $H^{\lambda}$, for $\lambda$ large. This section also contains a general outline of the proof of \cref{thm-split}. In \cref{sec:proof of the main theorem} we derive an exact reduction of the 
spectral problem for $H^{\lambda}$ acting in $L^2(\RR^2)$, for energies near the lowest two eigenvalues, to a problem on the subspace $\VV$. The asymptotic solution, for large $\lambda$,  of this reduced problem 
 requires a lower bound on the (complex) magnetic hopping coefficient $\rho^\lambda(d)$. This is proved in \cref{sec:hopping}. A number of the technical arguments in the proof of \cref{thm-split} are deferred to the appendices.

\subsection{Notation and conventions }\label{sec:notation}

\begin{enumerate}
\item In order to have less cluttered expressions in our calculations, we shall sometimes suppress the dependence of certain objects on the parameters: $\lambda$, $b=\lambda$ and $d$. We make these dependencies explicit in the statements of Theorems, Propositions etc. or when otherwise necessary.
\item $L^2=L^2(\RR^2)$ and $\ip{f}{g}=\ip{f}{g}_{L^2}$
\item $\mathcal{B}(X)$ is the space of bounded linear operators $X\to X$.
\item $h=h^{\lambda,b}$ is the single-well magnetic Hamiltonian (as defined in \cref{eq:the one-well Hamiltonian})\item $H^{\lambda,b}$ is the double-well magnetic Hamiltonian (as defined in \cref{eq:two-well Hamiltonian}) with constant magnetic field of amplitude $b$. 
\item For $b=\lambda$ we set $H^\lambda=H^{\lambda,b}$ and $E^\lambda=E^{\lambda,b}$.
\item $\vf_0$ is the ground state of $h$, $\vf_d$ is $\vf_0$ magnetically translated to $d$. Briefly we also use $\vf_m$ to denote the $m$th excited state of $h$, with $m=0,1,2,\dots$, hence strictly speaking $\vf_d$ is really $(\vf_0)_d$, but we avoid that notation and merely write $\vf_d$, since we never need $(\vf_m)_d$ for $m>0$.
\item For two functions of $\lambda$, $D(\lambda)$ and $ E(\lambda)$, we write $D(\lambda)\lesssim E(\lambda)$ (resp. $D(\lambda)\gtrsim E(\lambda)$) for $\lambda$ large if there is a constant $K$, which can be taken to be independent of $\lambda$ (for $\lambda$ sufficiently large), such that $D(\lambda)\ \le K\ E(\lambda)$ (resp.
$D(\lambda)\ \ge K\ E(\lambda)$) for 
 $\lambda\gg1$.
\end{enumerate}

\noindent\textbf{Acknowledgements:} We thank I. Corwin and Y. Lin for stimulating conversations. We also thank B. Helffer for his helpful comments on an earlier version of this article. C.L.F. was supported in part by National Science Foundation grant DMS-1700180. 
M.I.W. was supported in part by National Science Foundation grants DMS-1412560, DMS-1620418 and DMS-1908657 as well as by the Simons Foundation Math + X Investigator Award \#376319. 
J.S. acknowledges support by the Swiss National Science Foundation (grant number P2EZP2\_184228), as well as support from the Columbia University Mathematics Department and Simons Foundation Award \#376319, while a postdoctoral fellow
during 2018-2019. 
\bigskip

\section{The single-atom magnetic ground state}\label{sec:single-atom}

In this section we discuss properties of the magnetic ground state $\vf_0^{\lambda,b}(x)$, and its magnetic translate, $\vf^{\lambda,b}_{d}(x)$ (see below).
These states are used to define a subspace of {\it orbitals}, which for $\lambda$ large, approximately spans the low energy eigenspace of $H^{\lambda,b}$. \bigskip

Let $e_0^{\lambda,b}< e_1^{\lambda,b}\le e_2^{\lambda,b}\le\dots\le e_N^{\lambda,b}$ denote the first $N$ eigenvalues of $h^{\lambda,b}$ (as defined in \cref{eq:the one-well Hamiltonian}) listed with multiplicity, with corresponding eigenfunctions $\vf_m^{\lambda,b}$:

 \[ h^{\lambda,b}\vf_m^{\lambda,b}\ =\  e^{\lambda,b}\vf_m^{\lambda,b}, \quad  \|\vf_m^{\lambda,b}\|_{L^2}=1.\]
As stated, going forward, we follow the notational convention, for $b=\lambda$:  $h^\lambda=h^{\lambda,b}$, $e^{\lambda,b}=e^\lambda$ and $\vf_m^{\lambda,b}=\vf_m^{\lambda}$.

Under \cref{ass:radial monotone cpt supp for v}, the ground state eigenvalue, $e^{\lambda,b}_0$, is simple and its corresponding eigenfunction can be taken to be strictly positive: $\vf_0^{\lambda,b}>0$ and radial. 

We now reiterate the third point in \cref{ass:radial monotone cpt supp for v}: the spectral gap, $e_1^\lambda -e_0^\lambda$, is uniformly bounded away from zero for all $\lambda$ sufficiently large. Equivalently, we have \begin{assumption}\label{lem:single well gap}
Let $(\vf^{\lambda}_0,e^{\lambda}_0)$ denote the magnetic ground state eigenpair. 
There are constants $\lambda_\star, C_{\rm gap}>0$, which depend on $v_0$ and $d$, such that  
 if $\lambda>\lambda_\star$, then the following holds: For all  $\psi\in \mathrm{dom}(h^\lambda)$ such that $\ip{\vf_0^{\lambda}}{\psi}_{L^2}\ =\ 0$:
\begin{align}
\ip{\psi}{(h^{\lambda}-e^{\lambda}_0\Id)\psi}_{L^2}\ \geq C_{\rm gap}\  \norm{\psi}_{L^2}^2\ .
\label{gap}\end{align} 
\end{assumption}

\begin{rem} \cref{lem:single well gap} holds for potentials with a non-degenerate minimum \cite{Matsumoto_1994} where 
\begin{align} e_m^{\lambda}\ =\  -|v_{\rm min}|\ \lambda^2 + c_m \lambda + \mathcal{O}(1),\quad c_m>0,
\label{eq:gs-asym-alt}
\end{align}
and $c_m\sim m$ is proportional to the $m^{th}$ magnetic 2D harmonic oscillator state energy. Hence, 
in this case \cref{gap} is true also with the addition of a $\lambda$ factor on the RHS.
For piecewise constant potentials in constant magnetic fields, we \emph{expect} for all $\lambda\gg1$ that
 \begin{align} e_m^\lambda\  =\ -|v_{\rm min}|\ \lambda^2 + f_{m,\star} +  f_m(\lambda),\quad  f_{m\star}>0,\ \ f_m(\lambda)=o(1)\label{eq:gs-asym}
\end{align} 
and hence we would have $C_{\rm gap} \sim 1$; see \cite[Section 4.1, example (2)]{FLW17_doi:10.1002/cpa.21735} for the non-magnetic case.
\end{rem}

\subsection{Gaussian decay estimates for $\vf^{\lambda}_0(x)$ for $x$ outside of $\supp(v_0)$}

As discussed in the introduction, a key step in estimating the eigenvalue splitting (\cref{thm-split})
 is a proof of bounds on the magnetic hopping coefficient, $\rho^{\lambda}(d)$, defined in \eqref{m-hop}. Consider the integrand of the hopping coefficient 
  \eqref{m-hop}. 
  Since $\supp(v_0)\subset B_a(0)$, where $a< \frac12|d|$ (see \cref{eq:minimal spacing between wells}), for $x\in\supp(v_0)\subset B_a(0)$,
   we have $|x-d|>a$ and hence the function $x\mapsto \vf_0^{\lambda}(x-d)$ only samples values of $\vf_0^{\lambda}$ outside $B_a(0)$. In this region we shall make use of the following bounds on $\vf_0^{\lambda}$:

\begin{thm}\label{g-dec} There are constants $\mu_0,\mu_1, \lambda_\star>0$, which depend on $v_0$,  such that for all $\lambda>\lambda_\star$, we have the following:
The ground state $\vf^{\lambda}_0$ satisfies the Gaussian bounds:

	 \begin{align} 
	 \lambda^{-1}\exp\left(-\left(\frac14+\mu_0\right)\lambda \left(\left(|z|^2-a^2\right)+\mu_1\right)\right) \ \lesssim\ \vf_0^{\lambda}(z) \lesssim \sqrt{\lambda} \exp\left(-\frac14\lambda \left(|z|^2-a^2\right)\right) ,\quad {\rm for}\ |z|>a.
	 \label{eq:upper bound on wave function outside of well}\end{align}
\end{thm}

\noindent In the sequel, we use the intermediate estimate which plays a role in the proof of \cref{eq:upper bound on wave function outside of well}: \begin{align}\exp\left(-2\lambda(\vm/2)^{3/4}a^{-3/2}|z|^{2}\right)\lesssim \frac{\vf_0^\lambda(z)}{\frac{C_\lambda}{\sqrt{\lambda}}\exp\left(-\frac14\lambda |z|^2\right)}\lesssim 1\qquad(|z|>a)\,.\label{eq:upper bound on wave function outside of well keeping the normalization constant}\end{align} Here, $C_\lambda$, introduced just below in \cref{vf-out}, is a strictly positive $\lambda$-dependent constant (independent of $z$) whose upper and lower bounds, derived in the proof below, yield \cref{eq:upper bound on wave function outside of well}.

\begin{rem}
	We mention in passing the more general Gaussian decay results on eigenfunctions found in \cite{Erdos1996,Nakamura_doi:10.1080/03605309608821214}. There, using either the Feynman-Kac representation for $\vf^\lambda_m$ or PDE techniques respectively, one obtains Gaussian upper bounds (but not lower bounds) on $\vf^\lambda_m$, or any eigenfunction below the essential spectrum for potentials that decay at infinity but are not necessarily either compactly supported or radial.
\end{rem}

\begin{proof}[Proof of \cref{g-dec}] We write $\vf=\vf_0^{\lambda},v=v_0$ and $e=e_0^{\lambda}$ for brevity.
Since the ground state $\vf$ is radial (zero angular momentum), we have
\begin{eqnarray*}
	\left(-\Delta+\frac{1}{4}\lambda^{2}|x|^{2}+\lambda^{2}v\left(|x|\right)-e\Id\right)\vf & = & 0,\quad\ \vf>0,\quad 
	 \vf\in L^2(\RR^2).
\end{eqnarray*}
In  polar coordinates $\left(r,\theta\right)$, for
$r>a$:\begin{align}
	\left(-\partial_{r}^{2}-\frac{1}{r}\partial_{r}+\frac{1}{4}\lambda^{2}r^{2}-e\Id\right)\vf & = 0\,.
\label{rgea}\end{align}
Define
\begin{align*}
\vf\left(r\right) & =  e^{-\frac{1}{2}y}\psi\left(y\right),\quad \textrm{where}\ \ y =\frac{1}{2}\lambda r^{2}.
\end{align*} 
Then, \eqref{rgea} implies that $\psi(y)$ satisfies  Kummer's differential 
equation  
\begin{align}
	y\psi''\left(y\right)+\left(1-y\right)\psi'\left(y\right)-\alpha\psi\left(y\right) & =  0,\quad {\rm with}\ \alpha:=\frac{\lambda-e}{2\lambda}\label{kummer}\,.
\end{align}
Since asymptotically $e\approx -\lambda^2\vm$, we define \begin{align}\nu := \frac{\alpha}{\lambda}\,,\label{eq:def of v-star}\end{align} a positive quantity of order 1.

Equation \eqref{kummer} has two linearly independent solutions, $M_{\alpha}(y)$ (the Kummer
function) and $U_{\alpha}(y)$ (Tricomi's function). $M_\alpha(y)$ is regular at the origin but grows as $y$ tends to infinity (in a way that after multiplying by the Gaussian makes it non-normalizable),
while $U_{\alpha}(y)$ is singular at $y=0$ and decays as $y$ tends to infinity. Hence the ground state $\vf(r)$,  for $r>a$ must be a multiple of $U_{\alpha}\left(\frac{1}{2}\lambda r^{2}\right)$ and the Gaussian pre-factor.

An appropriate integral representation for
$U_{\alpha}(y)$ is given in \cite[pp. 129]{Temme_2014_asymp_methods_For_integrals}: 
\begin{eqnarray*}
	U_{\alpha}\left(y\right) & = & \frac{1}{\Gamma\left(\alpha\right)}\int_{0}^{\infty}\exp\left(-yt\right)t^{\alpha-1}\left(1+t\right)^{-\alpha}\dif{t}\qquad\left(\Re\left\{ \alpha\right\} >0\right)
\end{eqnarray*}
Hence we find that for $r>a$ (\emph{outside} the support of $v$), 
\begin{align}
	\vf_{\rm out}(r)\ \equiv\ \vf\left(r\right)\Big|_{r> a} & = C_\lambda\exp\left(-\frac{1}{4}\lambda r^{2}\right)\int_{0}^{\infty}\exp\left(-\frac{1}{2}\lambda r^{2}t\right)t^{\alpha-1}\left(1+t\right)^{-\alpha}\dif{t},
\label{vf-out}\end{align}
where the constant $C_\lambda>0$ is fixed by the normalization condition $\|\vf\|_{L^2(\RR^2)}=1$, determined jointly by \cref{vf-out} and  $\vf(r)$ for  $0\le r\le a$ (for which we have less information as it depends on the details of $v$). 
\begin{rem}We note in passing that for the Landau Hamiltonian ($V\equiv0$) in the symmetric gauge 
	\begin{equation}
	(P-\lambda A x)^2, \label{eq:Landau}
	\end{equation}
	the requirement of regularity of eigenstates at $r=0$ implies that any $L^2$ eigenstate with eigenvalue $e^\lambda_0$ is proportional to  $M_{\alpha}\left(\frac12 \lambda r^2\right)$,
	where $\alpha= (\lambda-e)/(2\lambda)$.
	Furthermore, the requirement of decay as  $r$ tends to infinity is also satisfied  if and only if  
	$\alpha$ is a non-positive integer.
	This implies the quantization condition of
	Landau levels: 
	$
	-m = \alpha\equiv (\lambda-e)/(2\lambda)$, where $m=0,1,2,\dots
	$. Equivalently,  	$e^\lambda_m  =  \lambda\left(2m+1\right)$, $m=0,1,2,\dots$.
\end{rem}
Resuming the proof of \cref{g-dec} now, before controlling $C_\lambda$, we obtain the dominant behavior of the integral in \cref{vf-out} for large $\lambda$.  
We rewrite that integral as
 $$ \int_{0}^{\infty}\frac{1}{t}\ee^{-\lambda\left(\frac{1}{2}r^{2}t+\vmh\log\left(1+\frac{1}{t}\right)\right)}\dif{t}\,, $$ where we have replaced $\alpha/\lambda$ by $\vmh$ according to \cref{eq:def of v-star}.

  The function within the exponential takes its minimum at \begin{align} t_{\star}\left(r\right):=\frac{1}{2}\left(\sqrt{1+\vmh\frac{8}{r^{2}}}-1\right)>0. \label{eq:definition of t-star}\end{align}
A judicious application of Laplace's approximation for the integral \cite[Section 2.4]{erdelyi1956asymptotic} yields
 \begin{align} \vf(r) \approx C_\lambda\sqrt{\frac{2\pi}{\lambda\vmh}\left(1+\frac{t_{\star}\left(r\right)^{2}}{1+2t_{\star}\left(r\right)}\right)}\exp\left(-\lambda\left(\frac{1}{4}(1+2t_{\star}\left(r\right))r^{2}+\vmh\log\left(1+\frac{1}{t_{\star}\left(r\right)}\right)\right)\right)\qquad(r>a)\,. \label{eq:exact asymptotic form of ground state outside the well not including asymptotics on the normalization constant}
 \end{align}
 
In particular for $\lambda$ sufficiently large we may derive simpler upper and lower bounds of $\vf$ from the right hand side of \cref{eq:exact asymptotic form of ground state outside the well not including asymptotics on the normalization constant}. For example, using that \begin{align*}\frac{1}{2}t_{\star}\left(r\right)r^{2}+\vmh\log\left(1+\frac{1}{t_{\star}\left(r\right)}\right)&=\vmh\left[\frac{1}{2}t_{\star}\left(r\right)\frac{r^{2}}{\vmh}+\log\left(1+\frac{1}{t_{\star}\left(r\right)}\right)\right] \\ &\leq\vmh 2\sqrt{\frac{r}{\sqrt{\vmh}}} = 2\vmh^{3/4}\sqrt{r}\end{align*} and $r\geq a$ (so $\sqrt{r}\leq a^{-3/2} r^2$) we obtain the lower bound
 \begin{align}\vf(r) \gtrsim \frac{C_\lambda}{\sqrt{\lambda}}\exp\left(-\lambda\left(\frac{1}{4}+2\vmh^{3/4}a^{-3/2}\right)r^{2}\right),\quad r>a\,,
 \label{simpler lower bound on ground state outside well not including normalization constant}\end{align}
and clearly we have the upper bound:
 \begin{align}\vf(r)\lesssim \frac{C_\lambda}{\sqrt{\lambda}}\exp\left(-\frac14\lambda r^2\right),\quad r>a\,.\label{eq:obvious upper bound on ground state}\end{align} This establishes \cref{eq:upper bound on wave function outside of well keeping the normalization constant}.

We next use these upper and lower bounds on $\vf(r)$ for $r>a$ to obtain upper and lower bounds on $C_\lambda$, using the normalization condition  \begin{align}
\norm{\varphi}_{B_{a}}^{2}+\norm{\varphi_{\rm out}}_{\mathbb{R}^{2}\setminus B_{a}}^{2} = 1 \label{eq:normalization constraint}\,.\end{align}

Let us begin with the upper bound on $C_\lambda$.
 For $\lambda\gg1$, the term 
$\norm{\varphi_{\rm out}}_{\mathbb{R}^{2}\setminus B_{a}}$ is asymptotically given by
 \begin{align} \norm{\varphi_{\rm out}}_{\mathbb{R}^{2}\setminus B_{a}}^{2} \approx \frac{2 \pi C_\lambda^2}{\lambda\vmh}\int_a^{\infty}r \left(1+\frac{t_{\star}\left(r\right)^{2}}{1+2t_{\star}\left(r\right)}\right) \exp\left(-2\lambda\left(\frac{1}{4}(1+2t_{\star}\left(r\right))r^{2}+\vmh\log\left(1+\frac{1}{t_{\star}\left(r\right)}\right)\right)\right) \dif{r}\,. \label{eq:normalization outside of the well asymptotic formula}\end{align}
 
We next asymptotically evaluate the integral in \cref{eq:normalization outside of the well asymptotic formula}. 
We note that the function $$ q(r) := \frac{1}{4}(1+2t_{\star}\left(r\right))r^{2}+\vmh\log\left(1+\frac{1}{t_{\star}\left(r\right)}\right) $$ is monotone increasing so its minimum is located at $r=a$, the endpoint of the integration interval. 
Therefore the asymptotic evaluation of 
the integral within \cref{eq:normalization outside of the well asymptotic formula} yields
 \begin{align} \norm{\varphi_{\rm out}}_{\mathbb{R}^{2}\setminus B_{a}}^{2} \approx \frac{2 \pi a C_\lambda^2}{\lambda^2\vmh q'(a)}  \left(1+\frac{t_{\star}\left(a\right)^{2}}{1+2t_{\star}\left(a\right)}\right) \exp\left(-2\lambda\left(\frac{1}{4}(1+2t_{\star}\left(a\right))a^{2}+\vmh\log\left(1+\frac{1}{t_{\star}\left(a\right)}\right)\right)\right).\label{eq:asymptotics on normalization integral outside the well}\end{align} 
 
Solving \cref{eq:asymptotics on normalization integral outside the well} for $C_\lambda$ and using that  $\norm{\varphi_{\rm out}}_{\mathbb{R}^{2}\setminus B_{a}}^{2} \leq 1 $ yields the upper bound: 
 \begin{align}C_\lambda \leq  \lambda\left(\frac{2 \pi a }{ \vmh q'(a)}\left(1+\frac{t_{\star}\left(a\right)^{2}}{1+2t_{\star}\left(a\right)}\right)\right)^{-1/2} \exp\left(+\lambda\left(\frac{1}{4}(1+2t_{\star}\left(a\right))a^{2}+\vmh\log\left(1+\frac{1}{t_{\star}\left(a\right)}\right)\right)\right)\,.
\label{eq:upper bound on normalization coefficient}\end{align}
Inserting this upper bound on $C_\lambda$ into \cref{eq:exact asymptotic form of ground state outside the well not including asymptotics on the normalization constant} (and not into the simpler \cref{eq:obvious upper bound on ground state}, as we require a tighter upper bound) yields
 \begin{align}\vf(r) \lesssim\sqrt{\lambda}  \frac{\exp\left(-\lambda\left(\frac{1}{4}(1+2t_{\star}\left(r\right))r^{2}+\vmh\log\left(1+\frac{1}{t_{\star}\left(r\right)}\right)\right)\right)}{\exp\left(-\lambda\left(\frac{1}{4}(1+2t_{\star}\left(a\right))a^{2}+\vmh\log\left(1+\frac{1}{t_{\star}\left(a\right)}\right)\right)\right)}\qquad(r>a)\,.\label{eq:full upper bound on ground state outside the well}\end{align}
Since both $r\mapsto t_\star(r)r^2$ and $r\mapsto \log(1+\frac{1}{t_\star(r)})$ are separately monotone increasing and $r>a$ we indeed get $$\vf(r) \lesssim \sqrt{\lambda} \exp(-\frac{1}{4}\lambda(r^2-a^2))$$ which is corresponds to the upper bound in \cref{eq:upper bound on wave function outside of well}.

To finish the proof of \cref{g-dec} we need to obtain a lower bound on the constant $C_\lambda$ to be used along with \cref{simpler lower bound on ground state outside well not including normalization constant}. To that end, let us define the function
 $$ E(r) := \lambda^2\vf^2(r)+(\vf^\prime(r))^2\qquad (r>0)\,. $$ 
 Differentiating $E$ and replacing  $\vf''$ by its expression in terms of $\vf$ in \cref{rgea} yields
  $$ E'(r) = \lambda\tilde{v}\left(r\right)2\lambda\varphi\left(r\right)\varphi'\left(r\right)-2\frac{1}{r}\varphi'\left(r\right)^{2} ,$$ 
where $\tilde{v}(r) := \frac{1}{4}r^{2}+v\left(r\right)-e/\lambda^2+1$. We have the basic bound $$|\tilde{v}(r)|\leq \frac14 R^2 +2\vm +1 =: A$$ for all $r\in[0,R]$ for all $R>0$, where we have used $|e|\leq \lambda^2\vm$.

Now fix $R\geq a$ and let $r\in\left(0,R\right]$. Then, we have
$E^\prime(r)\ge -\left(\lambda A + \frac{2}{r}\right)E(r) $. Integration over the interval $[r,R]$ implies 
\[ \lambda^2\vf^2(r)\ \le\ E(r) \leq R^{2}E\left(R\right)\exp\left(\lambda A\left(R-r\right)\right)r^{-2},\quad r\leq R. \]
Taking the square root and then multiplying by $r\vf(r)$ yields:
\[ \lambda r\ \vf^2(r)\ \leq R\sqrt{E\left(R\right)}\exp\left(\frac12\lambda A\left(R-r\right)\right)\vf(r),\quad r\leq R. \]
We have the bound  $\norm{\vf}_{L^{\infty}_{B_R}} \leq Q \lambda^2$ by \cref{intvfv} for some positive constant $Q$,
stated just below (we emphasize the proof of it does not rely on the lower bound on $C_\lambda$).
Hence, 

\[ \lambda r\ \vf^2(r)\ \leq Q \lambda^2 R\exp\left(\frac12\lambda A\left(R-r\right)\right)\sqrt{E\left(R\right)},\quad r\leq R. \]
Integration over $0\le r\le a$ gives a bound on $\|\vf\|^2_{L^2(B_a)}$, and since $\vf$ is normalized in $L^2(\RR^2)$:
\begin{align*}
1\ &\leq\ 
 \frac{4\pi Q R}{A}\ \exp\left(\frac12\lambda AR\right)\sqrt{E\left(R\right)}\ +\ 
 \|\vf_{\rm out}\|^2_{L^2(\RR^2\setminus B_a)} \end{align*}
 which by \cref{eq:asymptotics on normalization integral outside the well} gives:
 \[
1\ \lesssim\ \exp\left(\frac12\lambda AR\right)\ \sqrt{E\left(R\right)}\ +\ 
\frac{1}{\lambda^2}C_\lambda^2\exp(-\frac12\lambda a^2)\ .\]
 
We next obtain an upper bound for $E(R)=\lambda^2\vf_{\rm out}^2(R)+\vf_{\rm out}^\prime(R)^2$  in terms of $C_\lambda$ from \cref{eq:obvious upper bound on ground state}
 and a bound on $\vf_{\rm out}^\prime(R)$, which follows from the representation \cref{vf-out}: 
  $$ |\vf_{\rm out}'(R)| \lesssim \sqrt{\lambda} C_{\lambda}\ee^{-\frac{1}{4}\lambda R^{2}}\ .$$
Hence, for all $R\ge a$:
 $$E(R) \lesssim \lambda C_{\lambda}^2\ee^{-\frac{1}{2}\lambda R^{2}},  $$
 and so \begin{align*} \underbrace{\lambda^2\exp\left(+\frac12\lambda a^2\right)}_{=:p} &\lesssim \underbrace{\lambda^{2.5} \exp\left(\frac12\lambda \left(AR+a^2-\frac12 R^2\right)\right)}_{=:q} C_\lambda + C_\lambda^2
 \end{align*} which implies $$ C_\lambda \gtrsim \frac12 q\left(\sqrt{4\frac{p}{q^2}+1}-1\right)\,. $$
 Since we are free to choose any $R\geq a$, let us now specify to $R=a$. Then the term within the exponent in $\frac{p}{q^2}$, $\frac12(R^2-a^2)-AR=-Aa$, is always negative (we have $A>0$). As a result, for large $\lambda$, $4\frac{p}{q^2}\ll 1$ and we may estimate $$\sqrt{4\frac{p}{q^2}+1}-1\geq\frac{p}{q^2}$$ so that \begin{align} C_\lambda \gtrsim  \frac{p}{2q} &\sim \frac{1}{\sqrt{\lambda}} \exp\left(-\frac12\lambda a(A-a)\right)\nonumber\\
 &= \frac{1}{\sqrt{\lambda}} \exp\left(-\frac12\lambda a\left(\frac14(a-2)^2+2\vm\right)\right)\,.\label{clam-lb}
 \end{align}
  Therefore, inserting \cref{clam-lb} into  \cref{simpler lower bound on ground state outside well not including normalization constant} yields the lower bound in \cref{eq:upper bound on wave function outside of well}:
  \begin{align*} \vf(r) &\gtrsim \lambda^{-1}\exp\left[-\lambda\left(\frac{1}{4}+2(\vm/2)^{3/4}a^{-3/2}\right)r^{2}-\frac12\lambda a\left(\frac14(a-2)^2+2\vm\right)\right]\\
  &=:\lambda^{-1}\exp\left[-\lambda\left(\frac{1}{4}+\mu_0\right)\left(\left(r^{2}-a^2\right)+\mu_1\right)\right] \end{align*} with \begin{align}\mu_0 &= 2(\vm/2)^{3/4}a^{-3/2} \\
  \mu_1 &= \frac{\frac12 a\left(\frac14(a-2)^2+2\vm\right)}{\frac14(1+2(\vm/2)^{3/4}a^{-3/2})}+a^2 \,.\end{align}
\end{proof}

The above proof makes use of \cref{vf-in-bd} in

\begin{prop}\label{intvfv}
 Assume $a'>a$. There is a constant $C$, such that for all $\lambda$ sufficiently large:
	\begin{align}
	 \|\vf_0^\lambda\|_{L^\infty_{B_{a'}}} &\lesssim \lambda^2\label{vf-in-bd}\\
	 -\int_0^{a}\vf_0^\lambda\left(r\right)v_0\left(r\right)r\dif{r} &\geq C  > 0\ . 
	\label{eq:intvfv}
	\end{align}
\end{prop}

 The bound \cref{eq:intvfv} plays an important role in a lower 
bound on the hopping coefficient, see \cref{lb-rho}. The proof of  \cref{eq:intvfv} makes use of \cref{vf-in-bd}. With \cref{intvfv}, the proof of \cref{g-dec} is now complete.

\section{Magnetic orbital subspace $\mathcal{V}$,  the hopping coefficient and outline of the proof of \cref{thm-split}}\label{sec:magnetic orbital subspace}

\subsection{Magnetic translations}\label{sec:m-trans}

We recall that in the non-magnetic case, the low-lying (first two) eigenstates of the Hamiltonian in \cref{HnoB}
 are, for $\lambda$ large,  well-approximated by symmetric and anti-symmetric superpositions of 
states:
\begin{equation}
\vf_0(x)\quad {\rm and}\quad R^d\vf_0(x)\equiv e^{-i d\cdot P}\vf_0(x)=\vf_0(x-d).
\label{R-def}
\end{equation}
 Here, $\vf_0$ is the ground state
 for the atomic Hamiltonian $h^{\lambda,0}\equiv P^2+\lambda^2 v_0(x)$ and, by translation invariance of the Laplacian ($[R^d,P^2]=0$),
 $\vf(x-d)$ is the ground state of  $P^2+\lambda^2 v(x-d)$.
 
 In the case of constant nonzero magnetic field, we seek to approximate the low lying spectral subspace, in an analogous manner, by using ground states of the atomic \emph{magnetic} Hamiltonian. An essential difference here is that kinetic part of the magnetic Hamiltonian, $(P-b A x)^2$, does not commute with 
  translations ($[R^d,(P-b A x)^2]\ne0$). Here we have temporarily restored $b$ in order to emphasize the magnetic dependence.

 The situation is remedied by the use of  \emph{magnetic} translations, introduced by Zak \cite{Zak_1964_PhysRev.134.A1602}. Introduce the magnetic translation by $z\in\RR^2$:
  
\begin{equation}
 \hat{R}^z := e^{\ii x \cdot b A z}\ R^z\ =  e^{\ii x \cdot b A z}\ e^{-iz\cdot P}
 \end{equation}
  Then $[\hat{R}^z,(P-b A x)^2]=0$ for any $z\in\RR^2$. We therefore define
  the magnetic translation of the ground state centered at the origin $\vf\equiv\vf_0$, by the vector $d\in\RR^2$:
\[ \vf_d(x)\ \equiv\ \hat{R}^d \vf(x)\ =\ e^{ix\cdot \lambda A d}\ \vf_0(x-d), \] 
which satisfies the key relation:

\begin{align*}   
h^\lambda_d\ \vf_d(x)\ &=\ e_0^\lambda\ \vf_d(x),\quad \textrm{where}\\
h^\lambda_d\ &\equiv\ \hat{R}^d\ h^\lambda \hat{R}^{-d} =  (P- \lambda Ax)^2\ +\ \lambda^2 v_d(x)\ \ {\rm and}\\ 
v_\zeta(x) &=  R^{^\zeta} v_0(x)\ =\ v_0(x-\zeta),\quad \zeta\in\RR^2. 
\end{align*}

\subsection{Magnetic orbital subspace $\mathcal{V}={\rm span}\{\vf_0,\vf_d\}$}

The states $\vf_0$ and $\vf_d$ play the role of {\it magnetic atomic orbitals}. They satisfy
\begin{align} \HH^\lambda \vf_0^\lambda = v_d \vf_0^\lambda,\qquad  \HH^\lambda \vf_d^\lambda = v_0 \vf^\lambda_d,
\label{eq:app-estat}\end{align} 
where we have defined the centered about $e_0^\lambda$ two-well Hamiltonian
\[ \HH^\lambda :=\ H^\lambda\ -\ e^\lambda_0\Id\,.\]

For $\lambda$ large, the right hand sides of \cref{eq:app-estat} are bounded, uniformly in $x$, by  $e^{-C^\prime \lambda|d|^2}$, for some $C^\prime>0$, since
 $v_0$ is supported in $B_a(0)$ and $\vf_0$ satisfies the Gaussian bound of \cref{g-dec}.

We expect that for $\lambda$ large, the low lying spectrum of $H^\lambda$ to be determined by $H^\lambda$ restricted to the subspace of magnetic orbitals:
\begin{equation}
 \VV := \mathrm{span}\Set{\vf_0^\lambda,\vf_d^\lambda}\ .\label{VV-def}
 \end{equation}
Note that $L^2(\RR^2)\ =\ \VV\oplus\VV^\perp$ and introduce orthogonal projections:
\[ \Pi:L^2(\RR^2)\to \VV\quad {\rm and}\quad \Pi^\perp:L^2(\RR^2)\to \VV^\perp\ .\]

\noindent {\bf N.B.} To lighten up the notation, we shall often suppress the superscript $\lambda$: thus $\vf_0^\lambda\to\vf_0$ etc. 

\subsection{Strategy for the proof of \cref{thm-split}}\label{strategy}

We provide a detailed sketch of the proof of \cref{thm-split}.

\begin{enumerate}
\item[] Step 1: Via a Schur complement / Lyapunov-Schmidt reduction strategy we reduce the eigenvalue problem 
 $\HH\psi=z\psi$, $\psi\in L^2$ (|z| small--we recall that $\HH$ has been recentered at $e$) to a problem on the two dimensional subspace $\VV$:
\begin{equation}
 \Big[\ \Pi\ ({\HH}-z)\ \Pi\ + D^{\lambda}(z)\ \Big]\psi_\parallel\ =\ 0,\quad \psi_\parallel\in\VV,
 \label{red-V}
 \end{equation}
where  $z\mapsto  D^{\lambda}(z)$  is an analytic mapping from a neighborhood of $z=0$ into $\mathcal{B}(\VV)$.
This reduction uses the bounded invertibility of $\Pi^{\perp}(\HH-z)\Pi^{\perp}$. In particular, for $|z|$ small:
\begin{equation}
 \left\|\ \left(\Pi^{\perp}\ (\HH-z)\ \Pi^{\perp}\right)^{-1}\ \right\|_{\mathcal{B}(\VV^\perp)}\ \lesssim\ 1.
\label{res-perp}
\end{equation}
\item[] Step 2: For $\lambda$ large, we expand the operator in \eqref{red-V}. In terms of the basis
 $\Set{\vf_0,\vf_d}$, the condition for $z$ (with $|z|$ small) to be an eigenvalue of $\HH$ is equivalent to
\begin{equation}
\det \Big[\ \begin{pmatrix} -z&\rho^\lambda\left(d\right) \\ \overline{\rho^{\lambda}\left(d\right) }&-z\end{pmatrix} 
  \ +\ B^\lambda(z)\ \Big]=\ 0, 
 \label{exp-red-V}
 \end{equation}
where $B^\lambda(z)$ is a $2\times2$ matrix which depends analytically on $z$ near zero, and   $\rho^\lambda\left(d\right)$ is the double well magnetic hopping coefficient:
 \begin{equation} \rho^{\lambda}\left(d\right) \ =\lambda^2\
  \left\langle \vf_0, v_0 \vf_d\right\rangle\,;
 \label{m-hop1}
 \end{equation}
see \cref{m-hop}.  
Hence, $z$ with $|z|$ small, is an eigenvalue of $\HH$
  if and only if
  \begin{equation} z^2\ -\ |\rho^\lambda\left(d\right)|^2\ +\ F^\lambda(z)\ =\ 0, 
  \label{z-eqn}
  \end{equation}
  where $z\mapsto F^\lambda(z)$ is some analytic function in a neighborhood of the origin (see \cref{eq:def of F} below).
  \item[] Step 3: \cref{lb-hop} provides a lower bound on $|\rho^{\lambda}\left(d\right)|$ for $\lambda$ large. 
  Hence,  $z=|\rho^{\lambda}\left(d\right)|\ w$, with $|w|$ of order $1$, is an eigenvalue of $\HH$
  if and only if
  \begin{equation} w^2\ -\ 1\ +\ |\rho^\lambda\left(d\right)|^{-2}\ F^\lambda(\ |\rho^{\lambda}\left(d\right)|\ w\ )\ =\ 0, 
  \label{w-eqn}
  \end{equation}
  where $w\mapsto |\rho^{\lambda}\left(d\right)|^{-2}\ F^\lambda(\ |\rho^{\lambda}\left(d\right)|\ w\ )$ is analytic in a neighborhood of the origin. The bounds on $\rho^\lambda(d)$ asserted in \cref{lb-hop} allow us to control the correction term (the term involving $F$) in \cref{w-eqn}.
  
  Detailed estimates (\cref{sec:analysis of the various error terms in the reduction scheme}) imply: for any $\varepsilon>0$ small, there exists $\lambda_\star(\varepsilon)$ such  that for $\lambda>\lambda_\star$
  \begin{align} \Big|\ |\rho^{\lambda}\left(d\right)|^{-2}\ F^\lambda(\ |\rho^{\lambda}\left(d\right)|\ w\ )\ \Big|\ < |w^2-1|\label{eq:inequality for Rouche's theorem}\end{align}
  for all $w$ satisfying either $|w-1|=\varepsilon$ or $|w-(-1)|=\varepsilon$. It follows from Rouch\'e's theorem 
  that \eqref{w-eqn} has exactly one solution $w_-^\lambda$ satisfying $|w_-^\lambda -(-1)|<\varepsilon$ and 
  exactly one solution $w_+^\lambda$ satisfying $|w_+^\lambda -1|<\varepsilon$. Correspondingly, the \emph{non-centered} Hamiltonian \cref{eq:two-well Hamiltonian} has exactly two distinct eigenvalues $E_0^{\lambda}<E_1^{\lambda}$ with:
  \[ E_0^{\lambda}\ =\ e_0^{\lambda}\ +\ |\rho^{\lambda}\left(d\right)|\ w_-^\lambda\quad 
  \textrm{and}\quad E_1^{\lambda}\ =\ e_0^{\lambda}\ +\ |\rho^{\lambda}\left(d\right)|\ w_+^\lambda\ .\]
  Hence, the eigenvalue splitting--our main object of interest--is given by:
  \[ E_1^{\lambda}\ -\ E_0^{\lambda}\ =\ |\rho^{\lambda}\left(d\right)|\ \left[\ w_+^\lambda\ -\ w_-^\lambda\ \right]
  \ =\ 2\ |\rho^{\lambda}\left(d\right)|\ (\ 1+ \mathcal{O}(\varepsilon)\ ),\ \]
  for all $\lambda>\lambda_\star(\varepsilon)$. We have thus proven \cref{eq:key identity for eignevalue splitting}, and \cref{thm-split} follows from the upper and lower bounds
   on $|\rho^{\lambda}\left(d\right)|$ asserted in \cref{lb-hop} (with adjustment of the constants passing from \cref{rho-bd} to \cref{eq:main result}). 
  
\end{enumerate}

\section{Proof of \cref{thm-split}}
\label{sec:proof of the main theorem}
We carry out the proof of \cref{thm-split} following the strategy laid out in the previous section, with certain technical calculations left to appendices. 

\subsection{Step 1 of \cref{strategy}: Reduction to $\VV$}\label{sec:reduction to orbital subspace}

 Since $L^2=\VV\oplus\VV^\perp$, we  can identify $\psi\in L^2$ with the 
 pair $(\psi_\parallel,\psi_\perp)\in\VV\times\VV^\perp$. We write the eigenvalue problem $\HH\psi = z \psi$ as 
 \[ \begin{pmatrix} 
 \Pi\ (\HH-z)\ \Pi & \Pi\ \HH\ \Pi^\perp\\
 \Pi^\perp\ \HH\ \Pi & \Pi^\perp\ (\HH-z)\ \Pi^\perp
 \end{pmatrix} \begin{pmatrix} \psi_\parallel\\ \psi_\perp\end{pmatrix}\ =\  \begin{pmatrix} 0\\ 0 \end{pmatrix}\,.\
 \]
 The following proposition facilitates a reduction to a problem on $\VV$.
 
 \begin{prop}\label{prop:res-bd}
Fix $z$ near the origin. There exists $\lambda_\star$ such that, for all $\lambda>\lambda_\star$,  the  operator  
$\Pi^\perp\ ({\HH}^{\lambda}-z\Id)\ \Pi^\perp: \VV^\perp\to\VV^\perp$ is invertible, with the resolvent bound: 
 \begin{equation}
 \Big\|\ \left[\ \Pi^\perp\ ({\HH}^{\lambda}-z\Id)\ \Pi^\perp\right]^{-1}\ \Big\|_{\mathcal{B}(\VV^\perp)}\lesssim 1
 \label{eq:res-bd in early section}
 \end{equation}
 \end{prop}
 \cref{prop:res-bd} is proved via energy estimates in \cref{sec:en-est}. Using this proposition we may, by a Schur complement strategy, 
 eliminate $\psi_\perp$ and obtain a reduction to a problem on the two-dimensional subspace $\mathcal{V}$:
 \begin{equation}
\Big[\ \Pi\ ({\HH}^{\lambda}-z\Id)\ \Pi\ +\  D^\lambda(z)\ \ \Big]\ \psi_\parallel\ =\ 0\ .\label{psi-red}\end{equation}
Here $D^\lambda(z)$ is the rank two operator given by
\begin{equation}
D^\lambda(z)\ \equiv \ -\Pi\ {\HH^\lambda}\ \Pi^{\perp}\ \left(\Pi^{\perp}\ ({\HH^\lambda}-z\Id)\ \Pi^{\perp}\right)^{-1}\ \Pi^{\perp}\ {\HH^\lambda}\ \Pi\,.
\label{Cdef}
\end{equation}

\subsection{Step 2 of \cref{strategy} - expansion of \eqref{psi-red}}

 We choose an orthonormal basis for $\mathcal{V}$, 
$ \{\ \vf_0, \tilde{\vf}_d\ \}$, 
 where $\tilde{\vf}_d$ is obtained from $\vf_d$ by Gram-Schmidt orthonormalization: 
 \[  \tilde{\vf}_d \equiv \left(1-\left|\left\langle \vf_0,\vf_{d}\right\rangle \right|^{2}\right)^{-1/2}\Big(\ \vf_{d}-\left\langle \vf_0,\vf_{d}\right\rangle \vf_0\ \Big)\,. \]

With respect to this basis:
\[ \Pi\  (\HH^{\lambda}-z) \ \Pi	=	\begin{bmatrix}\left\langle \vf_0,  \HH^{\lambda} \vf_0\right\rangle\ -\ z  & \left\langle \vf_0, \HH^{\lambda}  \tilde{\vf}_{d}\right\rangle \\
	\left\langle \tilde{\vf}_{d}, \HH^{\lambda}  \vf_0\right\rangle  & \left\langle \tilde{\vf}_{d}, \HH^{\lambda}  \tilde{\vf}_{d}\right\rangle\ - \ z
\end{bmatrix} 
\] and similarly for $D(z)$.  Since $e^{-c'\lambda}\lesssim\left|\left\langle \vf_0,\vf_{d}\right\rangle\right| \lesssim e^{-c\lambda}$ for some $c,c'>0$, 
we introduce
\begin{equation}
 A^\lambda(d):=\begin{bmatrix}0 & \left\langle \vf_0,\HH^{\lambda}\vf_{d}\right\rangle \\
\left\langle \vf_{d},\HH^{\lambda}\vf_0\right\rangle  & 0
\end{bmatrix} = \begin{bmatrix}0 &  \rho(d) \\
\overline{\rho(d)}  & 0
\end{bmatrix} . \label{Ad}\end{equation}
The latter equality follows from \cref{eq:app-estat}. Therefore, 
\begin{align*}
\Pi\ ({\HH}^{\lambda}-z)\ \Pi\ +\  D^\lambda(z)\ =\ \begin{bmatrix}-z &  \rho(d) \\
\overline{\rho(d)}  & -z
\end{bmatrix}\ +\ B^\lambda(z)
\end{align*}
where
\begin{equation}
B^\lambda(z)\ \equiv \ \Pi \HH^{\lambda}\Pi -  A^\lambda(d)\ +\ D^\lambda(z)
\label{Bdef}\end{equation}
and $D^\lambda(z)$ is given by \cref{Cdef}.
Therefore, by the setup of \cref{sec:reduction to orbital subspace} we have

\begin{prop} For $z$ in a sufficiently small neighborhood of zero (uniformly in $\lambda$), 
$z$ is an eigenvalue of $\HH^\lambda$ (the Hamiltonian centered at $e_0^\lambda$) if and only if 
\begin{equation} \det\Big[\ \begin{bmatrix}-z &  \rho^\lambda(d) \\
\overline{\rho^\lambda(d)}  & -z
\end{bmatrix}\ +\ B^\lambda(z)\ \Big]\ =\ 0.
\label{det-z}\end{equation}

\end{prop}

We claim that for $z$ sufficiently near zero and $\lambda$ large, the $2\times2$ matrix-valued analytic function $B^\lambda(z)$ is a small correction and hence \cref{det-z} implies that $z^2-|\rho^\lambda(d)|^2\approx0$. To justify this we express $B^\lambda(z)$ in terms of the matrix elements of $H$ and show that their contribution is negligible. 

\subsection{Step 3 of \cref{strategy}--solving \cref{det-z} for all $z$ near zero}

In \cref{sec:hopping} we prove \cref{lb-hop} and in particular a lower bound on the magnitude of the hopping coefficient, $\rho^\lambda(d)$.
 With this in mind we write $w:=\left|\rho\left(d\right)\right|^{-1}z$. Then, 
 \cref{det-z} becomes
 \begin{align} w^{2}-1-f\left(w\right)w-g\left(w\right)=0, 
 \label{poly-w}\end{align}
 where
\begin{align}
 f\left(w\right)\ &:=\ \left|\rho\left(d\right)\right|^{-1}\tr\left(B\left(z\right)\right)\Big|_{z=\left|\rho\left(d\right)\right|w}
 \label{fw}\\
  g\left(w\right)\ &:=\ \left|\rho\left(d\right)\right|^{-2}\left(\det\left(B\left(z\right)\right)+ \left( \rho(d) B_{21}(z)\ +\ \overline{\rho(d)} B_{12}(z) \right)\ \right)\Big|_{z=\left|\rho\left(d\right)\right|  w} \nonumber\\
  &=\quad |\rho(d)|^{-2}\left(\det(B(z)) +
  \tr(\det(A)A^{-1}B(z))\right)\Big|_{z=\left|\rho\left(d\right)\right|  w}\,.
 \label{gw} \end{align}
   We thus identify $F$ from \cref{z-eqn} as \begin{equation}
   	|\rho\left(d\right)|^{-2}\ F^\lambda(\ |\rho\left(d\right)|\ w\ ) := -f\left(w\right)w-g\left(w\right)\,.\label{eq:def of F}
   \end{equation} 
 To bound $ |\rho\left(d\right)|^{-2}\ F^\lambda(\ |\rho\left(d\right)|\ w\ )$ we require
   \begin{prop}\label{B-bounds}
	For any fixed $K>0$, (with $\rho=\rho^\lambda(d)$ and $B(z)=B^\lambda(z)$), 
	\begin{align*}
		|\rho|^{-1}\tr(B(|\rho|w)),\quad 
		|\rho|^{-2}\det(B(|\rho|w)),\quad {\rm and}\quad 
		|\rho|^{-2}\tr(\det(A)A^{-1}B(|\rho|w))
	\end{align*}
	tend to zero as $\lambda\to\infty$ uniformly in $|w|\le K$.
\end{prop}
We prove \cref{B-bounds}  in \cref{sec:analysis of the various error terms in the reduction scheme}.  Now fix an arbitrary $\varepsilon>0$,
   then for $w\in\CC$ such that $|w^2-1|=\varepsilon$ we have:
   \begin{align}
    |\rho\left(d\right)|^{-2}\ \Big|\ F^\lambda(\ |\rho\left(d\right)|\ w\ )\ \Big| \ \le\
    |f\left(w\right)w|\ +\ |g\left(w\right)| < |w^2-1|,
    \end{align}
    for $\lambda>\lambda_\star(\varepsilon)$. This completes the proof of \cref{thm-split} modulo
     (i) bounds on the hopping coefficient in \cref{sec:hopping} and (ii) the proof of \cref{B-bounds}.

\section{The hopping coefficient}\label{sec:hopping}
In this and the following sections we shall often simplify expressions by writing $\vf$ for $\vf^\lambda_0$ and   $v$ for $v_0$. We have by \cref{eq:app-estat} and \cref{m-hop}: 
\begin{eqnarray*}
	\rho^\lambda\left(d\right) & \equiv & \left\langle \vf^\lambda_0,\HH^\lambda\vf^\lambda_{d}\right\rangle  =  \left\langle \vf_0, \lambda^2 v_{0}\vf_{d}\right\rangle \\
	& = & \lambda^{2}\int\overline{\vf\left(x\right)}v\left(x\right)\exp\left(\ii x\cdot\lambda Ad\right)\vf\left(x-d\right)\dif{x}\,.
\end{eqnarray*}

We first bound $|\rho^\lambda\left(d\right)|$ from above. We note that 
\[|\rho^\lambda\left(d\right)|\le \max_{x\in B_a(0)}|\vf(x-d)|\ \lambda^2\int_{B_a(0)} |v(x)\vf(x)|\dif{x}\ \lesssim \lambda^2\ \max_{x\in B_a(0)}|\vf(x-d)| \,.\]
For $x\in B_a(0)$ we have, since $a<\frac12|d|$, that $|x-d|>a$. Applying \cref{g-dec}
 to $\vf(x-d)$, we deduce the upper bound
\begin{align*}
	|\rho^\lambda(d)| 
	\lesssim\lambda^{2.5}\exp\Big(-\frac{1}{4}\lambda (\left(|d|-a\right)^2-a^2)\Big)\,.
\end{align*}

The proof of a lower bound is more difficult, due to the oscillatory character of the integral \cref{m-hop} defining $\rho^\lambda(d)$. Without any loss of generality we may assume $d= |d| e_{1}\equiv|d|(1,0)$ and we recall that $a<\frac12|d|$. Then,  $Ad=\frac12 e_3\wedge e_1= \frac12 e_2$ and hence the exponent in the integrand of \cref{m-hop} is
  $i x\cdot bAd=i\frac12 b |d| r\sin\theta=\frac{i}2 \lambda |d| r\sin\theta$ in polar coordinates. Thus,
\begin{align}
\rho^\lambda\left(d\right) & =\lambda^{2}\int_{\mathbb{R}^{2}}\overline{\vf\left(x\right)}v\left(x\right)
\exp\left(\ii x\cdot\lambda\frac{1}{2}|d| e_{2}\right)\vf\left(x- |d| e_{1}\right)\dif{x}\nonumber\\
& =\lambda^{2}\int_0^\infty\int_0^{2\pi}\vf\left(r\right)v\left(r\right)
\exp\left(\frac{\ii}2 \lambda |d| r\sin\theta\right)\vf\left(\sqrt{r^{2}+|d|^{2}-2r|d|cos\left(\theta\right)}\right)r\dif{r}\dif{\theta}
\nonumber\\
& =\lambda^{2}\int_0^a\vf\left(r\right)v\left(r\right)L_{|d|}\left(r\right)\dif{r} \nonumber\\
&\textrm{where}\quad  
L_{|d|}(r) := \int_0^{2\pi}\exp\left(\frac{i}2 \lambda |d| r\sin\theta\right)\vf^\lambda\left(\sqrt{r^{2}+|d|^{2}-2r|d|\cos\left(\theta\right)}\right)r\dif{\theta}
\label{Ld-def}\end{align}

We next study the integral kernel $L_{|d|}(r)$. 
Remarkably, although defined as an oscillatory integral, $L_{|d|}(r)$ can be shown to be positive
 and to have a representation as a \emph{non}-oscillatory integral.
 
 \begin{prop}\label{prop:ld-pos} Let $a<\frac{|d|}{2}$. Then, 
 \begin{align}
L_{|d|}(r) & =  C_\lambda\ r\ \exp\left(-\frac14\lambda(r^2+|d|^2)\right)\nonumber\\
&\qquad\qquad\times\ \int_0^\infty\ \exp\left(-\frac{\lambda}{2}(r^2+|d|^2)t\right)
	t^{\alpha-1}(1+t)^{-\alpha}\ I_0\left(\lambda|d|r\sqrt{t(t+1)}\right)\ dt
\label{Ld}\,	\end{align}
where $C_\lambda$ is the positive constant from \cref{eq:upper bound on wave function outside of well keeping the normalization constant}, $\alpha$ from \cref{kummer}, and $z\mapsto I_0(z)$ denotes the modified Bessel function of order zero,
which is strictly positive and grows exponentially for $z\in[0,\infty)$.
 \end{prop}
 
To prove \cref{prop:ld-pos} we use the following identity for modified Bessel functions:

\begin{lem}\label{lem:Bessel id}
	For any $\xi,\beta\in\RR$ with $\beta>\xi$ we have $$ \int_{0}^{2\pi}\exp(\ii \xi \sin\theta+\beta\cos\theta)\dif{\theta} = 2\pi I_0(\sqrt{\beta^2-\xi^2}) $$ where $I_0$ is the modified Bessel zeroth function, {\it i.e.} $I_0(x) \equiv J_0(\ii x)$.
\end{lem}
\begin{proof}
	We have \begin{align*}\ii \xi \sin(\theta)+\beta\cos(\theta)
	&\stackrel{\beta>\xi}{=}-\sqrt{\beta^2-\xi^2}\sin\left(\theta-\ii\left(\ii\frac{\pi}{2}+\frac12\log\left(\frac{\beta/\xi+1}{\beta/\xi-1}\right)\right)\right)\end{align*}
	Now for \emph{any} $\delta\in\CC$, $$ \partial_{\delta} \int_{0}^{2\pi} \exp(-x\sin(\theta-\ii\delta))\dif{\theta} = -\ii \int_{0}^{2\pi} \partial_{\theta}\exp(-x\sin(\theta-\ii\delta))\dif{\theta} = 0 \,, $$ the last equality follows by the fundamental theorem of calculus and the fact the integrand is $2\pi$-periodic in $\theta$. Hence combining these two equations we find that for any $\delta\in\CC$, $$ \int_{0}^{2\pi}\exp(\ii \xi \sin\theta+\beta\cos\theta)\dif{\theta} = \int_{0}^{2\pi}\exp\left(-\sqrt{\beta^2-\xi^2}\sin(\theta-\ii\delta)\right)\dif{\theta}$$ and so for appropriate choice of $\delta$ we find the result using the integral representation of $I_0$ \cite[Eq-n 9.6.16]{abramowitz+stegun}. 
\end{proof}

We  now turn to the proof of \cref{prop:ld-pos}.
Recall by \cref{eq:minimal spacing between wells} that $|x-d|>a$ for $x\in B_a(0)$.
Therefore by \cref{vf-out}), $\vf^\lambda(|x-d|)=\vf^\lambda_{\mathrm{out}}(|x-d|)$ for $|x-d|> a$.
Hence, 
\begin{align*}
	L_{|d|}(r) & = r \int_0^{2\pi}\exp\left(\frac{i}2 \lambda |d| r\sin\theta\right)\vf_{\rm out}^\lambda\left(\sqrt{r^{2}+|d|^{2}-2r|d|\cos\theta}\right)\ \dif{\theta}
\end{align*}
Substitution of the expression \cref{vf-out}  for $\vf^\lambda_{\mathrm{out}}(|x-d|)$ gives
\begin{align*}
L_{|d|}(r) 	& =  C_\lambda\ r\ \int_0^{2\pi}\exp\left(\frac{i}2 \lambda |d| r\sin\theta\right)\ \exp\left(-\frac{1}{4}\lambda\left(r^{2}+|d|^{2}-2r|d|\cos\theta \right)\right)\\
	&\qquad   \times\int_0^{\infty}\exp\left(-\frac{1}{2}\lambda\left(r^{2}+|d|^{2}-2r|d|\cos\theta\right)t\right)t^{\alpha-1}\left(1+t\right)^{-\alpha}\quad \dif{t}\ \dif{\theta}
	\end{align*}
	Combining and reorganizing exponents, and then interchanging order of integration yields:
	\begin{equation}
	L_{|d|}(r)  =  C_\lambda\ r\ \exp\left(-\frac14\lambda(r^2+|d|^2)\right)
	 \int_0^\infty t^{\alpha-1}(1+t)^{-\alpha}\exp\left(-\frac{\lambda}{2}(r^2+|d|^2)t\right)\ \mathcal{I}(t) \ \dif{t} ,
	\label{Ld1}\end{equation}
	where 
\begin{align} \mathcal{I}(t) = \int_0^{2\pi} \exp\left(\ \lambda|d|r\left(-\frac{i}{2}\sin\theta + \left(\frac12+t\right)\cos\theta\right)\ \right) d\theta\ =\ 2\pi I_0\left(\lambda|d|r\sqrt{t(t+1)}\right),\label{eq:the result of the Bessel integral}\end{align}
by an application of \cref{lem:Bessel id} with $\xi=\frac{\lambda}2|d|r$ and $\beta=\lambda|d|r(1/2+t)$. Using this relation for $\mathcal{I}(t,r,\lambda,|d|)$ in \cref{Ld1}, we conclude the representation for $L_{|d|}(r)$ in   \cref{Ld}. This completes the proof of \cref{prop:ld-pos}.

\subsection{Proof of the lower bound for $|\rho^\lambda(d)|$ in \cref{lb-hop}}\label{lb-rho}
The expression for  $L_{|d|}(r) $ in \cref{Ld} has non-negative integrand, so it may be bounded below as follows. The modified the Bessel function $z\mapsto I_0(z)$ has the following property \cite[Eq-n (9.3.14)]{Temme_2014_asymp_methods_For_integrals}: there are positive constants $C_1<C_2$, such that for all $z\ge0$
\begin{align} 
C_1\ \frac{\exp(z)}{\sqrt{2\pi z}+1}\  \leq\ I_0(z)\ \leq C_2\ \frac{\exp(z)}{\sqrt{2\pi z}+1}\,. \label{eq:estimate for modified Bessel}\end{align}

Hence, replacing that estimate in the integrand for $L_{|d|}(r)$ as well as using the basic estimates $\sqrt{t(t+1)}\geq t$ and $\log(\frac{t}{1+t})\geq-\frac{1}{t}$ we get the estimate (making use of \cref{eq:def of v-star}):
 $$ L_{|d|}(r) \geq C_1 C_\lambda r\ee^{-\frac{1}{4}\lambda\left(r^{2}+|d|^{2}\right)}\int_{0}^{\infty}\frac{\ee^{-\frac{1}{2}\lambda\left(|d|-r\right)^{2}t-\lambda\vmh\frac{1}{t}}}{t\sqrt{2\pi\lambda r|d|\sqrt{t\left(t+1\right)}}+t}\dif{t}\,. $$
 
A routine analysis involving Laplace's approximation \cite[Section 2.4]{erdelyi1956asymptotic} follows. The minimum of $ t\mapsto \frac{1}{2}\left(|d|-r\right)^{2}t+\vmh\frac{1}{t}$ is attained at $\frac{\sqrt{2\vmh}}{|d|-r}$, and so the integral is comparable to
$$ D  \frac{1}{\lambda}\exp\left(-\lambda\frac{\sqrt{2\vmh}}{2}\left(|d|-r\right)-\lambda\vmh\frac{|d|-r}{\sqrt{2\vmh}}\right)$$ for some constant $D$ that does not depend on $\lambda$.
 Putting the pieces together we find $$ L_{|d|}(r) \gtrsim C_\lambda \frac{1}{\lambda} r\exp\left(-\frac{1}{4}\lambda\left(|d|^{2}+r^{2}\right)-\lambda\sqrt{2\vmh}\left(|d|-r\right)\right)\, . $$

Finally we return to bounding $|\rho^\lambda\left(d\right)|$ from below. Using the lower bound, just above, on $L_{|d|}(r)$ and $v(r)\le0$, we have
\begin{align}
	|\rho^\lambda\left(d\right)| & \geq  \lambda^{2}\int_0^{a}\vf\left(r\right)\ \left(-v\left(r\right)\right)\ L_{|d|}(r) \dif{r}\nonumber\\
	& \gtrsim C_\lambda \frac{1}{\lambda} \lambda^{2}\int_0^{a}\vf\left(r\right)\ \left(-v\left(r\right)\right)\ r\exp\left(-\frac{1}{4}\lambda\left(|d|^{2}+r^{2}\right)-\lambda\sqrt{2\vmh}\left(|d|-r\right)\right)\dif{r}\nonumber\\
	& \gtrsim C_\lambda \frac{1}{\lambda} \exp\left(-\frac{1}{4}\lambda\left(|d|^{2}+4\sqrt{2\vmh}|d|+a^{2}\right)\right)\ \int_0^{a}\vf\left(r\right)\ \left(-v\left(r\right)\right)\ r\dif{r}\label{eq:lower bound on hopping coefficient with explicit normalization constant}\,.
	\end{align}

The lower bound in \cref{lb-hop} now follows via  \cref{intvfv} and \cref{clam-lb}.

\subsection{Relation between nearest neighbor and beyond nearest neighbor hopping coefficients;
 proof of \cref{comp-hop}}

As explained in the introduction, the following result is useful in analyzing the relation between 
crystalline systems with deep atomic potential wells subject to strong constant magnetic fields, and tight-binding discrete models. In \cite{Shapiro_Weinstein_2020} we follow this program.

\begin{prop}\label{lem:factorize NN out of NNN hopping}
We use the same convention that WLOG $d$ is chosen along the horizontal axis. Fix $\delta=|d|>2a$, i.e., $d=\delta e_1$. Then, there are constants $C_\star$ and $\lambda_\star$ such that for all $\lambda\geq\lambda_\star$ and all $\xi\ge0$:
	\begin{equation} |\rho^\lambda(\left(\delta+\xi\right)e_1)|\ \le\ C_\star\ e^{-\frac18\lambda(\xi^2+2\delta \xi)}\ |\rho^\lambda(\delta e_1)|.\label{rhodx}
	\end{equation}
\end{prop}
To prove \cref{comp-hop}, choose $\xi := (x-1)|d| = (x-1)\delta$.
	 
\begin{proof} Recall that $\rho(\delta e_1)=\lambda^{2}\int_0^\infty\vf\left(r\right) v\left(r\right) L_{\delta e_1}\left(r\right)\dif{r}$, where the kernel $L_{\delta e_1}\left(r\right)$ is given in \cref{Ld}. Now this kernel is composed of an exponential factor $ \exp\left(-\frac14\lambda(r^2+\delta^2)\right)$, which decreases with increasing $\delta$, multiplied by an integral. 
If there were only this exponential factor, we could complete the proof using its monotonicity and multiplicative properties. 
The integral factor however is not evidently a decreasing function of $\delta$ due to the modified Bessel function, $I_0$, appearing in the integrand.
Thus (slightly abusing the notation to write $\rho(\delta)$ instead of $\rho(\delta e_1)$) we split the exponential and rewrite $|\rho(\delta)|$ as:
\begin{equation} 
|\rho(\delta)|= \lambda^{2}\int_0^a \vf\left(r\right) |v\left(r\right)|\ \exp\left(-\frac18\lambda(r^2+\delta^2)\right)\ \tilde{L}_{\delta}\left(r\right)\dif{r},\label{rho2}
\end{equation}
where $\supp(v)\subset B_a(0)$ and for all $r\le a$:
\begin{equation}
\tilde{L}_{\delta}\left(r\right)\ =C_\lambda r\ 
 \int_0^\infty\ \exp\left(-\frac{\lambda}{2}(r^2+\delta^2)\left(t+\frac14\right)\right)
	t^{\alpha-1}(1+t)^{-\alpha}\ I_0\left(\lambda \delta r\sqrt{t(t+1)}\right)\ dt
\label{tLd}\, .\end{equation}

We now claim that there is a constant $C_\star$, such that  for all $0\le r\le a$ the function $\delta\mapsto \tilde{L}_{\delta}\left(r\right):(0,\infty)\to\RR$
 satisfies  
 \begin{equation}
 \tilde{L}_{\delta+\xi}\left(r\right)\le C_\star\ \tilde{L}_{\delta}\left(r\right).
 \label{tL-est}\end{equation}
Using \cref{tL-est} to bound the expression in \cref{rho2} implies the assertion of \cref{lem:factorize NN out of NNN hopping}.

So to conclude the proof of \cref{lem:factorize NN out of NNN hopping} we now verify \cref{tL-est}. Employing \cref{eq:estimate for modified Bessel} again, we have for $0\le r\le a$ and $\delta>0$:
\begin{align}
C_1\ \tilde{I}_\delta(r)\ &\leq \tilde{L}_\delta(r)\ \leq\ C_2\ \tilde{I}_\delta(r)\ , \quad \textrm{where}\ \label{Ld-bd} \\ 
	 \tilde{I}_\delta(r) &:=  C_\lambda\ r\ \int_0^\infty\ \frac{t^{\alpha-1}(1+t)^{-\alpha}}{\sqrt{2\pi\lambda\delta r\sqrt{t(t+1)}}+1}\exp\left(\lambda\phi(\delta, r,t)\right)
	\dif{t},\ \textrm{and}\nonumber\\
\phi(\delta,r,t) &:= -\frac{1}{2}(r^2+\delta^2)\left(t+\frac{1}{4}\right)+\delta r\sqrt{t(t+1)}
	\nonumber \end{align}

For any $r\in[0,a]$ and $t>0$, $\delta\mapsto \phi(\delta,r,t)$ is strictly decreasing. 
Indeed, note that  $ \phi'(\delta)  = -\delta \left(t+\frac14\right)+r\sqrt{t(t+1)}$. Moreover, by \cref{eq:minimal spacing between wells} we have  $a<\frac{\delta}2$ and hence $\frac{r}\delta\le\frac12$.

 Therefore, $$\frac{\phi^\prime(\delta)}{\delta}\le- \left(t+\frac14\right)+\frac 12 \sqrt{t(t+1)}\le -\frac{t}2 < 0.$$ for all $t>0$ and $r\in[0,a]$. Hence,   $\delta\mapsto\phi(\delta)$ is a strictly decreasing function for for all $t>0$ and $r\in[0,a]$. Since the prefactor of the exponential in the integrand of
  $\tilde{I}_\delta(r)$ is also decreasing as a function of $\delta$, we have that $\tilde{I}_\delta(r)$ is decreasing with $\delta$. Applying this monotonicity and \cref{Ld-bd} we have: for any $r\in[0,a]$ and all $\xi\ge0$:
	\begin{align*} 
	\tilde{L}_{\delta+\xi}(r) &\leq C_2\ \tilde{I}_{\delta+\xi}(r) \leq  C_2\ \tilde{I}_{\delta}(r)\ \leq  C_2\ C_1^{-1}\ \tilde{L}_{\delta}(r).
	\end{align*} 
	This concludes the proof of \cref{tL-est} with $C_\star=C_2\ C_1^{-1}$.
\end{proof}

\subsection{Proof of  \cref{intvfv}}\label{int-lb}

We shall write $B_r=B_r(0)$, since all discs in this proof will be centered at the origin.
	We seek a lower bound for $\int \vf |v| \dif{x} = 2\pi\int_0^{a}\vf\left(r\right) |v(r)| r\dif{r} $.
	
	 Fix $a_2>a_1>a$ and let $\Theta\in C^\infty(\RR^2\to[0,1])$ be a cut-off function supported in $B_{a_2}$ and equal to $1$ in $B_{a_1}$. Recall that  $\vf$ is a radial function and it also satisfies 
	$$ (P^2+\lambda^2 v(x)+\frac{1}{4}\lambda^2 |x|^2 -e\Id)\vf =0. $$
Taking the $L^2$ inner product  with $\Theta$ and noting that $\Theta\equiv1$ on the $\supp(v)$ yields 

	\begin{align}
	\begin{split}
	\lambda^2\ \int_{B_a} |v(x)|\vf(x) \dif{x} = 	\lambda^2\left\langle |v|\vf,\Theta\right\rangle&= \left\langle \vf,
		              P^2\Theta\right\rangle+\left\langle\left(\frac{1}{4}\lambda^2 |x|^2-e\Id\right)\vf,
		                     \Theta \right\rangle \,.
	\end{split}
	\label{eq:inner product of cut-off with eigen equation}
	\end{align}
	
To bound the right hand side, we note (i) $P^2\Theta$ is supported within the set $|x|\ge a_1>a$, where $\vf$ satisfies the pointwise exponentially small bound \cref{eq:upper bound on wave function outside of well} and that 
 on $\supp(\Theta) = B_{a_2}$ we have $\frac{1}{4}\lambda^2 x^2+|e^\lambda|\ge C\lambda^2$ by 
 \cref{eq:gs-asym}. Therefore
 \begin{align}
		\int_{\RR^2} |v|\vf & \geq\  C \lambda^2 \int_{B_{a_1}} \vf(x)\ dx\ -\  C^\prime \exp(-c \lambda)\,.
		\label{eq:lower bound on integral of potential and ground state in terms of ground state}
	\end{align}

The proof of \cref{intvfv} has now been reduced to a lower bound on  $\lambda^2 \int_{B_{a_1}} \vf$.

By normalization of $\vf$ and the pointwise Gaussian bound for $|x|>a$ (\cref{g-dec}), 
\begin{equation}
 1\ =\ \int_{\RR^2}\vf^2\ = \int_{B_{a_1}}\vf^2\ + \  \int_{\RR^2\setminus B_{a_1}}\vf^2\ \le\ \|\vf\|_{_{L^\infty(B_{a_1})}}\cdot \int_{B_{a_1}}\vf\ +\ C e^{-c\lambda}
\label{lb-nm}\end{equation}

To bound $\|\vf\|_{_{L^\infty(B_{a_1})}}$ we use the Sobolev-type bound:
 \begin{align}
		\norm{\vf}_{L^\infty(B_{a_1})} \leq C \left(\ \norm{P^2\vf}_{L^2(B_{a_2})} + \norm{\vf}_{L^2(B_{a_2})}\ \right)\,.
\label{Sobelev}	\end{align}
	Now $-P^2 \vf = (\lambda^2 v+\frac{1}{4}\lambda^2 |x|^2 -e\Id)\vf$, and for 
 $x\in B_{a_2}$ we have 
 \[\left(\lambda^2 v(x)+\frac{1}{4}\lambda^2 |x|^2 - e^\lambda\right)\vf(x) \leq C \lambda^2\vf(x).\]
 and  hence  $\norm{P^2\vf}_{L^2(B_{a_2})} \leq C \lambda^2$. It follows from \cref{Sobelev} that 
 \begin{align} \norm{\vf}_{L^\infty(B_{a_1})} \leq C \lambda^2\,. \label{eq:vf-infty}\end{align}
 This proves  \cref{vf-in-bd}, the first bound in \cref{intvfv}.
The proof of \cref{intvfv} is now completed by substituting \eqref{eq:vf-infty} into \cref{lb-nm} to obtain 
\cref{eq:intvfv}.

\appendix

\section{Proof of \cref{B-bounds}}\label{sec:analysis of the various error terms in the reduction scheme}
 
\subsection{Estimation of  $wf\left(w\right)\ =\ w\ \left|\rho\left(d\right)\right|^{-1}\tr\left(B\left(z\right)\right)\Big|_{z=\left|\rho\left(d\right)\right|w}$}\label{A1}

Recall the definition of the correction $2\times2$ matrix $B\left(z\right) \equiv \Pi \HH\Pi - A^\lambda(d)\ + D^\lambda(z)$, where 
\[ \Pi \HH\Pi - A^\lambda(d) =
 \begin{pmatrix} 
 \left\langle  \vf_0,\HH\vf_0\right\rangle & \left\langle  \vf_0,\HH(\tilde\vf_d-\vf_d)\right\rangle \\
 \left\langle (\tilde\vf_d-\vf_d),\HH  \vf_0\right\rangle & \left\langle  \tilde\vf_d,\HH\tilde\vf_d\right\rangle
 \end{pmatrix}\]
 and 
\[ D^{\lambda}(z) = -\Pi \HH\Pi^{\perp}\left(\Pi^{\perp}(\HH-z\Id)\Pi^{\perp}\right)^{-1}\Pi^{\perp}\HH\Pi. \]
We may bound the norm of $D^\lambda(z)$ via the resolvent estimate \cref{res-perp}:
\begin{align*}
\left\| D^\lambda(z) \right\| & \leq C'\norm{\Pi \HH\Pi^{\perp}\left(\Pi^{\perp}(\HH-z\Id)\Pi^{\perp}\right)^{-1}\Pi^{\perp}\HH\Pi}\ 
	\leq  C^{\prime\prime}\norm{\Pi \HH\Pi^{\perp}}^{2}\,.
\end{align*}
By linearity and continuity of the trace
\begin{align*}
	\left|\tr\left(B\left(z\right)\right)\right| & \leq \left|\left\langle \vf_0,\HH\vf_0\right\rangle \right|+\left|\left\langle \tilde{\vf}_{d},\HH\tilde{\vf}_{d}\right\rangle \right| + C^{\prime\prime}\norm{\Pi \HH\Pi^{\perp}}^{2}\ .
\end{align*}
Pick any constant, $K>1$, and let $z=\left|\rho\left(d\right)\right| w$, where $|w|\le K$. 
Therefore, for $\lambda\gg1$
\begin{equation}
 |f(w)| \le \left|\rho\left(d\right)\right|^{-1}\left|\left\langle \vf_0,\HH\vf_0\right\rangle \right| + 
 \left|\rho\left(d\right)\right|^{-1}\left|\left\langle \tilde{\vf}_{d},\HH\tilde{\vf}_{d}\right\rangle \right| + C^{\prime\prime}\left|\rho\left(d\right)\right|^{-1}\norm{\Pi \HH\Pi^{\perp}}^{2}
 \label{tr-bound}\end{equation}

We now bound each of the three terms in \cref{tr-bound} separately in the order that they appear.

\subsubsection{Bound on $ \left|\rho\left(d\right)\right|^{-1}\left|\left\langle \vf_0,\HH\vf_0\right\rangle \right|$} 

Using \cref{eq:app-estat} and \cref{eq:upper bound on wave function outside of well} and, by \cref{eq:minimal spacing between wells}
  that $|x-d|<a$ implies $|x|>a$:
\begin{align*}
	\left|\left\langle \vf_0,\HH\vf_0\right\rangle \right| & =  \left| \left\langle \vf_0,v_{d}\vf_0\right\rangle \right|\ \leq 
	 \lambda^{2}\ \int_{B_a(d)}\left|\vf_0\right|^{2}\lesssim\lambda^2 \max_{z\in B_a(d)}|\vf(x)|\times\max_{z\in B_a(d)}|\vf(x)|\,.
\end{align*}
For one of these factors, we use the upper bound in \cref{eq:upper bound on wave function outside of well keeping the normalization constant} (keeping $C_\lambda$ explicit) and for the other one we use the upper bound in \cref{eq:upper bound on wave function outside of well}. This yields $$ \left|\left\langle \vf_0,\HH\vf_0\right\rangle \right|\lesssim\lambda^2 \times C_\lambda \frac{1}{\sqrt{\lambda}}\exp\left(-\frac14\lambda(|d|-a)^2\right)\times \sqrt{\lambda}\exp\left(-\frac14\lambda\left(\left(|d|-a\right)^2-a^2\right)\right) $$ 
where $C_\lambda$ is the particular constant from \cref{eq:upper bound on wave function outside of well keeping the normalization constant}. 

Using the lower bound on $\rho^\lambda(d)$ of \cref{lb-hop}, or more specifically keeping $C_\lambda$ explicit (so, really using \cref{eq:lower bound on hopping coefficient with explicit normalization constant}) so that it cancels with the same factor of $C_\lambda$ above we have for $\lambda\gg1$:
\begin{align}  \left|\rho\left(d\right)\right|^{-1}\left|\left\langle \vf_0,\HH\vf_0\right\rangle \right| \ &\lesssim \lambda^2 \frac1\lambda \exp\left(-\frac{1}{2}\lambda\left(|d|-a\right)^{2}+\frac14 \lambda a^2+\frac14\lambda\left(|d|^2+4\sqrt{2\vmh}|d|+a^2\right)\right)\nonumber\\
&=\lambda  \exp\left(-\frac{1}{4}\lambda|d|^2+(\sqrt{2\vmh}+a)\lambda|d|\right)\,.
\label{term1}\end{align}
Hence as long as $$|d| > 4(\sqrt{2\vmh}+a)$$ and $\lambda$ is sufficiently large, $\left|\rho\left(d\right)\right|^{-1}\left|\left\langle \vf_0,\HH\vf_0\right\rangle \right|$ decays exponentially to zero as $\lambda\to\infty$ and we have our result. Using \cref{eq:def of v-star}, this is asymptotically \cref{eq:spacing}.

Similarly, $ \left|\rho\left(d\right)\right|^{-1}\left|\left\langle \vf_d,\HH\vf_d\right\rangle \right|\to0$ as $\lambda\to\infty$. 

\subsubsection{Bound on $\left|\rho\left(d\right)\right|^{-1}\left|\left\langle \tilde{\vf}_{d},\HH\tilde{\vf}_{d}\right\rangle \right|$}

We recall that  $
	\tilde{\vf}_{d}  \equiv  \left(1-\left|\left\langle \vf,\vf_{d}\right\rangle \right|^{2}\right)^{-1}\left(\vf_{d}-\left\langle \vf,\vf_{d}\right\rangle \vf\right)$. By our pointwise decay estimates on $\vf_0$, we have 
	$\left|\left\langle \vf_0,\vf_d\right\rangle\right|\lesssim e^{-c\lambda}$ and therefore  
	 \begin{align*}
	\left|\left\langle \tilde{\vf}_{d},\HH\tilde{\vf}_{d}\right\rangle \right| 
	& \lesssim \left(1+e^{-c\lambda}\right)\left(\left|\left\langle \vf_{d},\HH\vf_{d}\right\rangle \right|
	+e^{-c\lambda}\left|\left\langle \vf_{d},\HH\vf_0\right\rangle \right|+e^{-c\lambda}\left|\left\langle \vf_0,\HH\vf_0\right\rangle \right|\right)
\end{align*}
Using  \cref{eq:app-estat}, note that
$\left|\left\langle \vf_{d},\HH\vf_0\right\rangle \right|=\left|\rho^\lambda(d) \right|$. Therefore, 
multiplication by  $\left|\rho\left(d\right)\right|^{-1}$ and using \cref{term1} we conclude that 
$\left|\rho\left(d\right)\right|^{-1}\left|\left\langle \tilde{\vf}_{d},\HH\tilde{\vf}_{d}\right\rangle \right|\to0$ as $\lambda\to\infty$. 

\subsubsection{Bound on $\norm{\Pi \HH\Pi^{\perp}}_{\mathcal{B}(L^2)}^{2}$}
\label{subsubsec:controlling the non-linear term}
We may equivalently bound the adjoint $\Pi^\perp \HH\Pi$. Using the relations \cref{eq:app-estat} and that $\left|\left\langle \vf_0,\vf_d\right\rangle\right|\lesssim e^{-c\lambda}$, we have
 \begin{align*}
 \HH\Pi &= \left\langle\vf_0,\cdot\right\rangle \HH\vf_0+\left\langle\tilde{\vf}_d,\cdot\right\rangle \HH\tilde{\vf}_d = \left(\left\langle\vf_0,\cdot\right\rangle +\mathcal{O}(e^{-c\lambda})\left\langle\tilde{\vf}_d,\cdot\right\rangle\right)v_d\vf_0
  +\left\langle\tilde{\vf}_d,\cdot\right\rangle v_0 {\vf}_d\,.
  \end{align*}
  Therefore, $\|\HH\Pi\|_{\mathcal{B}(L^2)} \lesssim \|v_d\vf_0\|_{L^2} + 
\| v_0 {\vf}_d\|_{L^2} \lesssim \left| \left\langle \vf_0,v_d\vf_0\right\rangle\right| = \left|\left\langle \vf_0,\HH\vf_0\right\rangle \right|$. Furthermore, $\|\Pi^\perp \HH\|_{\mathcal{B}(L^2)}$ satisfies the same bound.  Finally, multiplication by $\left|\rho\left(d\right)\right|^{-1}$ and using \cref{term1} we have that 
$|\rho^{-1}|\norm{\Pi \HH\Pi^{\perp}}_{\mathcal{B}(L^2)}^{2}\to0$ as $\lambda\to\infty$.

\subsection{Estimation of $g(w)$, given by \cref{gw}}

We need to bound the expression
\begin{equation} g(w)=\left(\det\left(\ \left|\rho\left(d\right)\right|^{-1}B\left(z\right)\right)+ 
\left( \left|\rho\left(d\right)\right|^{-1}\rho(d) B_{21}(z)\ +\ \left|\rho\left(d\right)\right|^{-1}\overline{\rho(d)} B_{12}(z) \right)\ \right)\Big|_{z=\left|\rho\left(d\right)\right|  w},
\label{gw1}
\end{equation}
where $B(z)=\Pi \HH \Pi - A^\lambda(d) + D^\lambda(z)$ was already studied in \cref{A1}. The latter two terms in \cref{gw1} tend to zero as $\lambda\to\infty$ by very similar estimates to those above. Indeed, we bound $ \left|\det\left(\left|\rho\left(d\right)\right|^{-1}B\left(z\right)\right)\right|$, using that the determinant of a $2\times2$ matrix  is bounded in terms of the square of its largest matrix element.
The magnitude of the diagonal elements of  $\Pi \HH \Pi -A$ and the $\|D^\lambda(z)\|$ have already been controlled above.
It remains to bound the magnitude of the off-diagonal terms elements of $\Pi \HH \Pi -A$. A representative bound, using $\left\langle\vf_0,\vf_d\right\rangle\lesssim e^{-c\lambda}$, is
$
|\ip{\vf}{\HH(\tilde{\vf}_d-\vf_0)}| \lesssim e^{-c\lambda}\ \left( \left|\ip{\vf_0}{\HH\vf_0}\right|\ +\ 
|\ip{\vf_0}{\HH\vf_d}|  \right),$
which tends to zero as $\lambda\to\infty$ by earlier bounds.

This completes the proof of \cref{B-bounds}.

\section{Resolvent bound: $\left(\ \Pi^\perp\ ({\HH}^{\lambda}-z\Id)\ \Pi^\perp\ \right)^{-1}\in \mathcal{B(\VV^\perp)}$ via energy estimates--  proof of  \cref{prop:res-bd}}
\label{sec:en-est}

Our goal is to prove that for all  $z$ sufficiently near  the origin

\begin{equation}
 \Big\|\ \left[\ \Pi^\perp\ ({\HH}^{\lambda}-z\Id)\ \Pi^\perp\right]^{-1}\ \Big\|_{\mathcal{B}(\VV^\perp)}\lesssim 1
 \label{eq:res-bd}\,.
 \end{equation}
  
At the heart of the proof of \cref{eq:res-bd} is an energy estimate for $\HH^\lambda$ restricted to $\VV^\perp$,
 which we now explain. Recall that $h^\lambda$ denotes the atomic magnetic Hamiltonian, \eqref{eq:the one-well Hamiltonian}, with $b=\lambda$,
 and $(e_0^\lambda,\vf_0)$ its normalized ground state eigenpair: $(h^\lambda-e_0^\lambda)\vf^\lambda_0=0,\ \vf^\lambda_0>0$. Furthermore, $h_d^\lambda=\hat{R}^dh^\lambda (\hat{R}^{d})^\ast$ denotes the magnetic translation of $h^\lambda$;  $\vf^\lambda_d=\hat{R}^d\vf^\lambda_0$, the magnetically translated ground state of $h^\lambda$; and $v_d(x)=v(x-d)$ the 
 atomic potential, $v_0$, re-centered at $x=d$. 
 The key to \cref{eq:res-bd} is:

\begin{thm}\label{thm:en-est}
There exist positive constants $\lambda_\star, C_1$ and $C_2$ such that for all $\lambda>\lambda_\star$ the following holds:
 Suppose $\psi\in H^2(\RR^2)$ with $\psi\in\VV^\perp=\{\vf_0,\vf_d\}^\perp$, {\it i.e.}
 \[ \left\langle\vf_0,\psi\right\rangle = 0\quad {\rm and}\quad \left\langle\vf_d,\psi\right\rangle = 0. \]
 Then, 
 \begin{equation}
 \left\langle \psi,\HH^\lambda\psi\right\rangle \ge C_1 \|\psi\|^2\ +\ C_2\lambda^{-2}\ \|(P-\lambda Ax)\psi\|^2
 \label{energy-est}\end{equation}
 Here, $C_1,C_2\in\mathcal{O}(1)$.
\end{thm}

\cref{thm:en-est} implies that $\Pi^\perp\ {\HH}^{\lambda}\ \Pi^\perp$ 
 and hence $\Pi^\perp\ ({H}^{\lambda}-z\Id)\ \Pi^\perp$, for $z$ sufficiently small,  are invertible on $L^2(\RR^2)$ 
 and satisfy \cref{eq:res-bd}.

 The point of departure of the proof of \cref{thm:en-est} is the following energy estimate, a simple extension of  \cref{lem:single well gap}:
\begin{lem}\label{lem:at-en-est}
For all $\lambda$ sufficiently large, if $\psi\in H^2(\RR^2)$ such that 
 $\left\langle \vf^\lambda_0,\psi\right\rangle=0$, then 
\begin{equation}
 \left\langle \psi,(h^\lambda-e^\lambda_0)\psi\right\rangle \ge C_{\rm gap} \|\psi\|^2,
 \label{at-en-est}\end{equation}
 where $C_{\rm gap} \sim 1$. Analogously,  if $\psi\in H^2(\RR^2)$ is such that 
 $\left\langle \vf^\lambda_d,\psi\right\rangle=0$, then 
 \begin{equation}
 \left\langle \psi,(h_d^\lambda-e^\lambda_0)\psi\right\rangle \ge C_{\rm gap} \|\psi\|^2,
 \label{tat-en-est}\end{equation}
 \end{lem}

 The strategy is now to build up the energy estimate for the magnetic double well operator, $H^\lambda$, from
  the atomic energy estimates for $h^\lambda$ and $h^\lambda_d$ by a spatial-localization strategy. An arbitrary $\psi\in H^2(\RR^2)$ is, via a spatial partition of unity, decomposed
  as  a superposition of: 
  (i) a function concentrated mainly within the support of the atomic potential, $v_0$, (ii) a function concentrated mainly within the support of the translated atomic potential, $v_d$, and
(iii) a function which is concentrated outside the supports of $v_0$ and $v_d$. 

Evaluating $\left\langle \psi,\HH^\lambda\psi\right\rangle$ on $\psi$, decomposed in this manner, gives
 rise to terms which are controlled either by (i) an energy estimate of type \cref{at-en-est} for functions concentrated
  on the support of $v_0$, (ii) energy estimate  of type \cref{tat-en-est} for functions concentrated
  on the support of $v_d$
  or (iii) an energy estimate for the shifted Landau Hamiltonian $(P-\lambda Ax)^2 - e^\lambda\approx (P-\lambda Ax)^2 + |v_{min}|\lambda^2$ for functions concentrated
  on the complement of the supports of $v_0$ and $v_d$. 
  
  We refer to \cite[Section 9, Lemma 9.1 through Proposition 9.4]{FLW17_doi:10.1002/cpa.21735} for detailed proofs in the context of the strong binding
  regime for graphene-like periodic (non-magnetic) Schr\"odinger operators and to \cite[Section 5]{Shapiro_Weinstein_2020} for a crystal in a constant magnetic field.
  
   The necessary adaptations to handle
  the constant magnetic field are not difficult to implement. In particular, the proofs require a magnetic version of integration by parts (IMS localization formula; see e.g. \cite{AIHPA_1983__38_3_295_0}):
 \begin{lem}\label{lem:magnetic IMS}
 		We have for any $\psi,\eta$, $$ (P-\lambda A x)^2 \eta\psi = \eta (P-\lambda A x)^2\psi +\psi P^2 \eta + 2 (P\eta) \cdot (P-\lambda A x) \psi\,.   $$
 	\end{lem}

 \bigskip

\begingroup
\let\itshape\upshape
\printbibliography
\endgroup
\end{document}